\documentclass [journal,onecolumn,11pt]{IEEEtran}
\usepackage{amsfonts,amsmath,amssymb}
\usepackage{indentfirst, setspace}
\usepackage{url,float}
\usepackage{color}
\usepackage[a4paper,ignoreall]{geometry}
\usepackage{longtable}
\usepackage{graphicx}
\usepackage[a4paper,ignoreall]{geometry}
\usepackage{multicol}
\usepackage{stfloats}
\usepackage{enumerate}
\usepackage{cite}
\usepackage{amsthm}
\usepackage{multirow}
\usepackage{amssymb}
\usepackage[all]{xy}
\usepackage[square, comma, sort&compress, numbers]{natbib}

\def\qu#1 {\fbox {\footnote {\ }}\ \footnotetext { From Qu: {\color{red}#1}}}
\def\hqu#1 {}
\def\li#1 {\fbox {\footnote {\ }}\ \footnotetext { From Li: {\color{blue}#1}}}
\def\hyin#1 {}

\newtheorem{Th}{Theorem}[section]
\newtheorem{Cor}[Th]{Corollary}

\newtheorem{Lemma}[Th]{Lemma}

\newtheorem{example}{Example}

\newtheorem{Rem}[Th]{Remark}
\newtheorem{Alg}[Th]{Algorithm}

\newcommand{\tr}{{\rm Tr}}

\newcommand{\gf}{{\mathbb F}}

\makeatletter
\newcommand{\figcaption}{\def\@captype{figure}\caption}
\newcommand{\tabcaption}{\def\@captype{table}\caption}
\makeatother

\begin{document}
	\title{New Constructions of Permutation Polynomials of the Form  $x^rh\left(x^{q-1}\right)$ over $\gf_{q^2}$}	
	\author{{Kangquan Li, Longjiang Qu and Qiang Wang}
		\thanks{Kangquan Li and Longjiang Qu are with the College of Science,
			National University of Defense Technology, Changsha, 410073, China.
			Qiang Wang is with School of Mathematics and Statistics, 			Carleton University, 1125 Colonel By Drive, Ottawa, Ontario, K1S 5B6,
			Canada.			
			The research of Longjiang Qu is partially supported by the National Basic Research Program of China (Grant No. 2013CB338002),
			the Nature Science Foundation of China (NSFC) under Grant 61272484, 11531002, 61572026,
			the Program for New Century Excellent Talents in University (NCET)
			and the Basic Research Fund of National University of Defense Technology (No.CJ 13-02-01). 
			The research of Qiang Wang is partially supported by NSERC of Canada.
			
			E-mail: likangquan11@nudt.edu.cn, ljqu\_happy@hotmail.com, wang@math.carleton.ca.		
			
			Corresponding author: Longjiang Qu\newline}}
	\maketitle{}

\begin{abstract}
	Permutation polynomials over finite fields have been studied extensively recently due to  their wide applications in cryptography, coding theory, communication theory,  among others.  
	Recently, several authors have studied permutation trinomials of the form $x^rh\left(x^{q-1}\right)$ over $\gf_{q^2}$, where $q=2^k$, $h(x)=1+x^s+x^t$ and $r, s, t, k>0$ are integers.
     Their methods are  essentially usage of a multiplicative version of AGW Criterion  because  they all transformed the problem of proving permutation polynomials over $\gf_{q^2}$ into that of showing the corresponding fractional polynomials permute a smaller set $\mu_{q+1}$, where $\mu_{q+1}:=\{x\in\mathbb{F}_{q^2} : x^{q+1}=1\}$.
	 Motivated by these results, we characterize the permutation polynomials of the form  $x^rh\left(x^{q-1}\right)$ over $\gf_{q^2}$ such that $h(x)\in\gf_q[x]$ is arbitrary and  $q$ is also an arbitrary prime power. Using AGW Criterion twice, one is multiplicative and the other is additive, we reduce the problem of proving permutation polynomials over $\gf_{q^2}$ into that  of showing permutations over a  small subset $S$ of a proper subfield $\gf_{q}$, which is significantly different from previously known methods.  In particular,  we demonstrate our method by constructing many new explicit classes of permutation polynomials of the form $x^rh\left(x^{q-1}\right)$ over $\gf_{q^2}$. {Moreover, we can explain most of the known permutation trinomials, which are in \cite{DQ,LiNian1,LikangquanFFa,LikangquanConstructed,Zhazhengbang,RS}, over finite field with even characteristic}. 
\end{abstract}


\section{Introduction}

Let $\gf_q$ be the finite field with $q =p^k$ elements. A polynomial $f\in\gf_q[x]$ is called a \emph{permutation polynomial} (PP) if the induced mapping $x\to f(x)$ is a permutation of $\gf_q$. The study of permutation polynomials over finite fields attracts a lot of interest for many years due to their wide applications in coding theory \cite{LC,ST,CD}, cryptography \cite{RSL} and combinatorial designs  \cite{DJ}. For example,  interesting cycle codes from several classes of permutation monomials and trinomials are constructed by Ding  in \cite{CD}.  Moreover,   a family of skew Hadamard difference sets via the Dickson permutation polynomial of order five  are constructed by Ding and Yuan in \cite{DJ}.  The latter discovery disproved the longstanding conjecture on skew Hadamard difference sets. 

In 2011, Akbary, Ghioca and Wang \cite{AGW} gave the following result, which is called AGW Criterion. 


\begin{Th}
	(\cite{AGW}, AGW Criterion)
	Let $A, S$ and $\overline{S}$ be finite sets with $\# S=\# \overline{S}$, and let $f: A\to A,$ $h: S\to \overline{S}$, $\lambda: A\to S$ and $\overline{\lambda}: A\to\overline{S}$ be maps such that $\bar{\lambda}\circ f=h\circ \lambda$. If both $\lambda$ and $\bar{\lambda}$ are surjective, then the following statements are equivalent:
	\begin{enumerate}[(i)]
		\item $f$ is a bijection; and
		\item $h$ is a bijection from $S$ to $\overline{S}$ and $f$ is injective on $\lambda^{-1}(s)$ for each $s\in S$.
	\end{enumerate}
\end{Th}
AGW Criterion can be represented by the following simple diagram.
\begin{equation*}
\xymatrix{
	A \ar[rr]^{f}\ar[d]_{\lambda} &   &  A  \ar[d]^{\overline{\lambda}} \\
	S	 \ar[rr]^{h} &  & \overline{S} }
\end{equation*}

The importance of the AGW Criterion depends on that it can be used not only to explain some previous constructions of PPs, but also to construct numerous new classes. For example,  Akbary, Ghoica and Wang \cite{AGW} applied their approach into different cases (i.e., multiplicative group case, elliptic curve case, additive group case) and obtained many interesting results. Furthermore, Yuan and Ding \cite{YuanDing} obtained descriptions of some permutation polynomials with special forms, such as $f(x)=g(B(x))+\sum_{i=1}^{r}\left(L_i(x)+\delta_i\right)h_i(B(x))$, $p(x)=f(x)g(\lambda(x))$ and so on. More classes of PPs of the form $L(x) +  g(x^q -x + \delta) \in \mathbb{F}_{q^n}[x]$, where $L(x) $ is a linearized polynomial and  $g(x)^q = g(x)$ were given in \cite{YuanDing:14}.   Recently, by employing the AGW Criterion two times, Zheng, Yuan and Pei \cite{Zheng} found a series of simple conditions for $$f(x)=\left(ax^q+bx+c\right)^r\phi\left(\left(ax^q+bx+c\right)^{\left.\left(q^2-1\right)\middle/d\right.}+ux^q+vx\right)\in\gf_{q^2}[x]$$ to permute $\gf_{q^2}$. 
Readers can consult \cite{XH1} for a recent survey on constructions of permutation polynomials.

 Permutation trinomials over finite fields are in particular interesting for their simple algebraic forms and additional extraordinary properties. For instances, Dobbertin \cite{HD} proved that the power function $x^{2^m+3}$ on $\gf_{2^{2m+1}}$ is an {APN} function and  the key of his proof was the discovery of a class of permutation trinomials. The discovery of {another} class of permutation trinomials by Ball and Zieve \cite{SM} provided a way to prove the construction of the Ree-Tits symplectic spreads of $\mathbf{PG}(3,q)$. Hou \cite{XH4} acquired a necessary and sufficient condition about determining a special permutation trinomial ($ax+bx^q+x^{2q-1}$ over $\gf_{q^2}$) through the Hermite Criterion.  In 2015, Ding et al. \cite{DQ}, presented a few class of permutation trinomials over finite fields with even characteristic. Since then, there have been increasingly attention on constructing permutation trinomials over finite fields, in particular, over finite fields with even characteristic. 

Recently, several authors \cite{LiNian1,LiNian2, LikangquanConstructed,LikangquanFFa,Zhazhengbang,RS} constructed permutation trinomials of the form  $x^rh\left(x^{q-1}\right)$ over $\gf_{q^2}$, where $q=2^k$, $h(x)=1+x^s+x^t$ and $r, s, t, k>0$ are integers. The main methods they used were similar, depending on the following criterion about permutation polynomials of the form {$x^rh\left(x^{(q-1)/d}\right)$} over $\gf_q$. In fact, the AGW Criterion is a generalization of the following lemma obtained by several authors \cite{PL,WQ,Zieve}.
\begin{Lemma}
	\cite{PL,WQ,Zieve}
	\label{lem1}
	Pick $d,r > 0$ with $d\mid (q-1)$, and let $h(x)\in\mathbb{F}_q[x]$. Then $f(x)=x^rh\left(x^{\left.(q-1)\middle/d\right.}\right)$ permutes $\gf_q$ if and only if both
	\begin{enumerate}[(1)]
		\item $\mathrm{gcd}\left(r,\left.(q-1)\middle/d\right.\right)=1$ and
		\item $x^rh(x)^{\left.(q-1)\middle/d\right.}$ permutes  $\mu_d :=\{x\in\mathbb{F}_q : x^d=1\}$. 
	\end{enumerate}
\end{Lemma}
According to Lemma \ref{lem1}, the problem of proving permutation trinomials $f(x)=x^rh\left(x^{q-1}\right)$ over $\gf_{q^2}$ is transformed into that of showing $g(x)=x^rh(x)^{q-1}$ permutes $\mu_{q+1}$. In this paper, $g(x)$   is called as {\emph{the corresponding fractional polynomials} } of $f(x)$. Generally, verifying that $g(x)$  permutes $\mu_{q+1}$  is also difficult. However, when the degree is not high  and the number of terms in the corresponding fractional polynomial $g(x)$ is not large, this problem may be solved 
by proving that $g(x)\neq g(y)$ for any $x\neq y\in\mu_{q+1}$. For convenience, we call this method as \emph{the fractional approach}. However, with the increment of the degree or the number of terms in the corresponding fractional polynomial, this problem becomes more difficult.  

Motivated by these constructions  \cite{DQ,LiNian1,LiNian2,LikangquanCCDS,LikangquanConstructed,LikangquanFFa,Zhazhengbang,RS}, we want to unify and generalize those results as widely as possible. In this paper, we characterize permutation polynomials of the form  $f(x)=x^rh\left(x^{q-1}\right)\in\gf_q[x]$ over $\gf_{q^2}$,  such as $h(x) \in \mathbb{F}_q[x]$ is arbitrary,  regardless of the value of characteristic of $\gf_{q^2}$.  We also provide a general way to construct permutation polynomials of these forms and demonstrate our method by constructing many explicit classes of permutation polynomials.

Let $\tr$ denote the trace function from $\gf_q$ to $\gf_p$ and  $\eta$ be a quadratic character over $\gf_{q}$ throughout this paper. Let 
\begin{equation}
\label{S}
S :=\begin{cases}
\{a\in\gf_q^{*}:\tr\left(\frac{1}{a}\right)=1\} & \text{ if $\mathrm{char}\gf_{q^2}$ is even},\\
\{a\in\gf_q:\eta\left(a^2-4\right)=-1\} & \text{ if $\mathrm{char}\gf_{q^2}$ is odd}.
\end{cases}
\end{equation}
The key point of our new method  is the following commutative diagram which is based on the AGW Criterion. 
\begin{equation}
\label{Diagram}
\begin{split}
\xymatrix{
	\gf_{q^2} \ar[rr]^{f(x)} \ar[d]_{x^{q-1}} &  &  \gf_{q^2} \ar[d]^{x^{q-1}}\\
	\mu_{q+1} \ar[rr]^{g(x)}\ar[d]_{x+x^q} &   & \mu_{q+1}  \ar[d]^{x+x^q} \\
	\{2,-2\}\cup S	 \ar[rr]^{R(x)} &  & \{2,-2\} \cup S } 
\end{split}
\end{equation}
We note that when $q$ is even, $\{2, -2\} \cup S$ is actually $\{0\} \cup S$.   Here we  employ the AGW criterion twice. The first one is the application of Lemma \ref{lem1}, which reduces the permutation of $\gf_{q^2}$ to a permutation of $\mu_{q+1} := \{ x \in \gf_{q^2}:  x^{q+1} =1 \}$.  On the basis of this, we use AGW criterion again, transforming the problem of determining whether $g(x)=x^rh(x)^{q-1}$ permutes $\mu_{q+1}$ into that of verifying whether a rational function $R(x)$ (defined later) permutes a subset of $\gf_q$.  Our method is different from the fractional approach. Not only the size of subset that we consider is further reduced (from $q+1$ to $\lceil\frac{q-1}{2}\rceil$), but also the subset can be contained in the subfield $\gf_{q}$. We demonstrate that a permutation of such subset can be easily verified in many occasions.  As consequences, we can reconstruct several known results easily and also construct many new explicit classes of permutation polynomials with more terms over $\gf_{q^2}$.

The remainder of this paper is organized as follows.  In Section 2, we explain our method to characterize or to construct permutation polynomials of the form  $x^rh\left(x^{q-1}\right)$ over $\gf_{q^2}$. Then we demonstrate our method by constructing many explicit classes of permutation polynomials in Section 3 and 4 respectively, when the characteristic of $\gf_{q^2}$ is even and odd respectively.  In the case of even characteristic, we reduce the problem into a problem of constructing rational functions $L(b) = b + l(b) + l(b)^2$ that permutes $T := \{ b\in \gf_q: \tr(b) = 1\}$. In particular, if $l(b)$ is a linearized polynomial, then the latter is equivalent to constructing rational functions $L(b) = b + l(b) + l(b)^2$ that permutes the subfield $\gf_q$.   In the case of odd characteristic, we focus on rational functions $L(b) = bl(b)^2$ permuting 
$T :=\{ b \in \gf_q : \eta(b) = -1, \eta(b+4)=1 \}$, where $\eta$ is a quadratic character over $\gf_{q}$.  Many explicit classes of permutation polynomials using simple choices of $l(b)$'s are constructed in both cases. 
In Section 5, we use this new method  to derive several recent constructions of permutation trinomials over finite fields with even characteristic. It is quite interesting that most of the known classes of permutation trinomials in \cite{DQ,LiNian1,LikangquanFFa,LikangquanConstructed,Zhazhengbang,RS} can be obtained in this way  (see TABLE V).  Finally Section 6 is the conclusion.


\section{A New method of constructing permutation polynomials over $\gf_{q^2}$}

Let $q=p^k$, where $p$ is a prime, and $f(x)=x^rh\left(x^{q-1}\right)\in\gf_q[x]$. In this section, we consider the permutation property of $f(x)$ over $\gf_{q^2}$. According to Lemma \ref{lem1}, it suffices to consider whether $g(x)$ permutes $\mu_{q+1}$. For $x\in\mu_{q+1}\backslash\{1,-1\},$ let $a=x+x^q=x+x^{-1}\in\gf_q$ and $S:=\{a=x+x^{-1}: x\in\mu_{q+1}\backslash\{1,-1 \}\}$. In the following, we determine the set $S$. On one hand, when $\mathrm{char}\gf_{q^2}=2$,  we know that the equation $x^2+ax+1=0$ has no solution in $\gf_{q}$ since $x\in\mu_{q+1}\backslash\{1\}$. Due to the following lemma, we have that $\tr\left(\frac{1}{a}\right)=1$. Therefore, in this case, $S:=\{a\in\gf_q^{*}:\tr\left(\frac{1}{a}\right)=1\}$.

\begin{Lemma}
	\cite{LN}
	\label{lem2}
	Let $q=2^k$, where $k$ is a positive integer. The quadratic equation $x^2+ux+v=0$, where $u,v\in\gf_{q}$ and $u\neq 0$, has roots in $\gf_q$ if and only if $\tr\left(\frac{v}{u^2}\right)=0$.
\end{Lemma}

 On the other hand, when $\mathrm{char}\gf_{q^2}$ is odd, we know that the equation $x^2-ax+1=0$ has no solution in $\gf_q$, neither. Then $\Delta=a^2-4$ is a nonsquare in $\gf_q$, i.e., $\eta\left(a^2-4\right)=-1$, where $\eta(\cdot)$ is  the quadratic residue. Hence, we have

\begin{equation}
\label{S}
S :=\begin{cases}
\{a\in\gf_q^{*}:\tr\left(\frac{1}{a}\right)=1\} & \text{ if $\mathrm{char}\gf_{q^2}$ is even},\\
\{a\in\gf_q:\eta\left(a^2-4\right)=-1\} & \text{ if $\mathrm{char}\gf_{q^2}$ is odd}.
\end{cases}
\end{equation}

According to the diagram (\ref{Diagram}), for verifying the permutation property of $g(x)$ over $\mu_{q+1}$, the key point is to prove the permutation property of $R(a)$ over $S$. But it is not trivial to obtain $R(a)$ from $g(x)$. We need the following algorithm to reduce the degree of $h(x)$.  Then we derive $R(a)$ from $h(x)$.

\begin{Alg}
	\label{alg1} 
    Let $x\in\mu_{q+1}$ and  $a=x+x^{-1}$. 
    \begin{enumerate}[1:]
    	\item {\bfseries Input} $h(x)$;
    	\item Plugging $x^2=ax-1$ into $h(x)$ and simplifying $h(x)$;
    	\item while $\mathrm{deg}h(x)>1$ then
    	\item \quad repeat \emph{2};
    	\item end while;
    	\item {\bfseries Output} the coefficients $h_1(a),h_2(a)$  of $h(x)$;
    \end{enumerate}
\end{Alg}
 
\begin{example}
 Let $\mathrm{char}\gf_q=2$ and $h(x)=1+x^2+x^{-1}$. Then $h(x)=1+ax+1+a+x=(a+1)x+a$. Therefore, $h_1(a)=a+1$ and $h_2(a)=a$ in the case.
\end{example}

{\begin{Rem}
		Let $n>0$ be an integer.  Assume $x^n=\phi_n(a) x+\chi_n(a)$, where $\phi_n(a), \chi_n(a)\in\gf_q[a]$ according to Algorithm \ref{alg1}.  Then it is trivial that $\phi_1(a)=1, \phi_2(a)=a$, $\chi_1(a)=0$ and $\chi_2(a)=-1$. Moreover, $x^{n+1}=\phi_{n+1}(a)x+\chi_{n+1}(a)=\left(a\phi_n(a)+\chi_n(a)\right)x-\phi_n(a).$ Therefore, $\phi_{n+1}(a)=a\phi_n(a)+\chi_n(a)$ and $\chi_{n+1}(a)=-\phi_n(a)$. 
\end{Rem}}

	Let $e\ge 1$ be an integer. It is trivial that $x^e+x^{-e}$ can be expressed as a function of $a$,  denoted by $D_e(a)$, through several iterations, where $a=x+x^{-1}$.  We remark $D_e(a)$ is known as Dickson polynomial of the first kind.  For example, $D_3(a)=a^3-3a$. And for any $e$, it is clear that $D_e(2)=2$.  The following result is the main theorem in this paper. 

\begin{Th}
	\label{Maintheorem}
	Let $f(x)=x^rh\left(x^{q-1}\right)\in\gf_{q^2}[x]$ such that all coefficients of $h(x)$ belong to $\gf_q$ and $S$ be the set defined in (\ref{S}). Let $a= x+x^{-1}$.  By Algorithm \ref{alg1}, we get $h_1(a)$ and $h_2(a)$ from $h(x)=h_1(a)x+h_2(a)$. Assume that $$R(a)=\frac{h_1^2(a)D_{r-2}(a)+h_2^2(a)D_r(a)+2h_1(a)h_2(a)D_{r-1}(a)}{h_1^2(a)+h_1(a)h_2(a)a+h_2^2(a)}.$$ Then $f(x)$ permutes $\gf_{q^2}$ if and only if the following conditions hold simultaneously:
	\begin{enumerate}[(i)]
		\item $\gcd(r,q-1)=1$;
		\item for the corresponding fractional polynomial $g(x)=x^rh(x)^{q-1}$, $g(x)=1$ has only one solution $x=1$ in $\mu_{q+1}$ and $g(x)=-1$ has only one solution $x=-1$ in $\mu_{q+1}$;
	    \item $h(x)\neq0$ for any $x\in\mu_{q+1}$;
		\item $R(a)$ permutes $\{2,-2\}\cup S$.
	\end{enumerate}
\end{Th}

\begin{proof}
	Firstly, let us recall the following diagram.  
	\begin{equation*}
	\xymatrix{
		\gf_{q^2} \ar[rr]^{f(x)} \ar[d]_{x^{q-1}} &  &  \gf_{q^2} \ar[d]^{x^{q-1}}\\
		\mu_{q+1} \ar[rr]^{g(x)}\ar[d]_{x+x^q} &   & \mu_{q+1}  \ar[d]^{x+x^q} \\
		\{2,-2\}\cup S	 \ar[rr]^{R(a)} &  & \{2,-2\} \cup S }
	\end{equation*}
	
	\emph{(1) The reduction from $\gf_{q^2}$ to $\mu_{q+1}$.}
	
	 According to Lemma \ref{lem1}, we know that $f(x)=x^rh\left(x^{q-1}\right)$ is a permutation polynomial over $\gf_{q^2}$ if and only if $\gcd(r,q-1)=1$ and $g(x)=x^rh(x)^{q-1}$ permutes $\mu_{q+1}$.
	
	\emph{(2) The reduction from $\mu_{q+1}$ to $\{2,-2\}\cup S$.}
	
	 It is clear that $h(x)\neq0$ for any $x\in\mu_{q+1}$ if $g(x)$ permutes $\mu_{q+1}$.  In the following, we claim that $R\circ \left(x+x^q\right)=\left(x+x^q\right) \circ g(x)$, i.e., $R(a)=g(x)+g(x)^q$ for $x\in\mu_{q+1}$. 
	 Firstly, we show that for $x=1,-1$, the above equation holds. When $x=1$, $a=2$ and $R(2)=\frac{2\left(h_1^2(2)+h_2^2(2)+2h_1(2)h_2(2)\right)}{h_1^2(2)+h_2^2(2)+2h_1(2)h_2(2)}=2$. On the other hand, since $g(1)=h(1)^{q-1}=1$ due to $h(x)\in\gf_{q}[x]$ and $h(1)\neq0$, so $g(1)+g(1)^q=2=R(2)$. Therefore, when $x=1$, $R(a)=g(x)+g(x)^q$. The case $x=-1$ is similar because 
$g(-1) = (-1)^r h(-1)^{q-1} = -1$ when $q$ is odd.   
	 Now, we consider the case $x\in\mu_{q+1}\backslash\{1,-1\}$, i.e., $a=x+x^{-1}\in S$ according to the definition of $S$. In fact, 
	\begin{eqnarray*}
	g(x)+g(x)^q &=& x^rh(x)^{q-1}+x^{-r}h(x)^{1-q} \\
	&=& \frac{x^{2r}h(x)^{2q}+h(x)^2}{x^rh(x)^{q+1}}.
	\end{eqnarray*}
	From $h(x)=h_1(a)x+h_2(a)$ obtained by Algorithm \ref{alg1}, we have 
	\begin{eqnarray*}
	g(x)+g(x)^q &=& \frac{x^{2r}\left(h_1(a)x^{-1}+h_2(a)\right)^2+\left(h_1(a)x+h_2(a)\right)^2}{x^r\left(h_1(a)x^{-1}+h_2(a)\right)\left(h_1(a)x+h_2(a)\right)}\\
	&=& \frac{x^{2r}\left(h_1^2(a)x^{-2}+h_2^2(a)+2h_1(a)h_2(a)x^{-1}\right)+h_1^2(a)x^2+h_2^2(a)+2h_1(a)h_2(a)x}{x^r\left(h_1^2(a)+h_1(a)h_2(a)x^{-1}+h_1(a)h_2(a)x+h_2^2(a)\right)}\\
	&=& \frac{x^{2r-2}h_1^2(a)+x^{2r}h_2^2(a)+2x^{2r-1}h_1(a)h_2(a)+x^2h_1^2(a)+h_2^2(a)+2h_1(a)h_2(a)x}{x^r\left(h_1^2(a)+h_1(a)h_2(a)a+h_2^2(a)\right)}\\
	&=&\frac{h_1^2(a)\left(x^{r-2}+x^{2-r}\right)+h_2^2(a)\left(x^r+x^{-r}\right)+2h_1(a)h_2(a)\left(x^{r-1}+x^{1-r}\right)}{h_1^2(a)+h_1(a)h_2(a)a+h_2^2(a)}
	\end{eqnarray*}
which is $R(a)$. Therefore, the above diagram is commutative. 

  Then according to the AGW Criterion, we have that $g(x)$ permutes $\mu_{q+1}$ if and only if $R(a)$ permutes $\{2,-2\}\cup S$ and for any $a\in\{2,-2\}\cup S$, $g(x)$ is injective on $V_a :=\{x\in\mu_{q+1}: x+x^q=a\}$. Since $|V_{2}|$ and $|V_{-2}|$ are both $1$, it suffices to prove that $g(x)\neq g\left(x^q\right)$ for any $x\in\mu_{q+1}\backslash\{1,-1\}$. Actually, $g\left(x^q\right)|_{\mu_{q+1}}=x^{-r}\frac{h(x^q)^{q}}{h(x^q)} = x^{-r} \frac{h(x^{-1})^q}{h(x^{-1})}= x^{-r} =g(x)^{-1}$  because $h(x)\neq0$ for any $x\in \mu_{q+1}$. If there exists $x\in\mu_{q+1}\backslash\{1,-1\}$ such that $g(x)=g\left(x^q\right)$, then $g(x)^2=1$, i.e., $g(x)=1$ or $g(x)=-1$. Therefore, the condition that $g(x)\neq g\left(x^q\right)$ for any $x\in\mu_{q+1}\backslash\{1,-1\}$ is equivalent to that $g(x)=1$ has only one solution $x=1$ in $\mu_{q+1}$ and $g(x)=-1$ has only one solution $x=-1$ in $\mu_{q+1}$. Therefore the proof is complete.
\end{proof}

By Theorem \ref{Maintheorem}, the problem of determining permutation polynomials over $\gf_{q^2}$ is transformed into that of considering permutations over $S$, which is different from the fractional approach. Moreover, we can construct a lot of permutation polynomials of the form  $x^rh\left(x^{q-1}\right)$ over $\gf_{q^2}$ from $R(a)$ which permutes $S$. Therefore, given $R(a)$, it is necessary to have an algorithm to obtain $h(x)$, which is not trivial. Firstly, we consider how to obtain $h(x)$ when we have $h_1(a)$ and $h_2(a)$. The method to obtain $h_1(a)$ and $h_2(a)$ will be introduced in Section 3 and Section 4 respectively depending on the parity of $\mathrm{char}\gf_{q^2}$. First of all, we can obtain an original $h(x)=h_1(a)x+h_2(a)$ according to the relationship $a=x+x^{-1}$. But what should be noticed is that the original $h(x)$ may be divisible by $x-x^{-1}$, which means $h(1)=0$, $h(-1)=0$ and hence breaks Condition $(iii)$ of Theorem \ref{Maintheorem}. Therefore, after obtaining the original $h(x)$, we repetitively compute $\frac{h(x)}{x-x^{-1}}$  until the new $h(x)$ satisfies that $h(1)\neq0$ and $h(-1)\neq0$.

\begin{Alg}
	\label{alg}
	Given $h_1(a), h_2(a)$.
	\begin{enumerate}[1:]
		\item {\bfseries Input} $h_1(a), h_2(a)$; 
		\item Computing $h(x)=h_1\left(x+x^{-1}\right)x+h_2\left(x+x^{-1}\right)$;
		\item while $h(1)=0,h(-1)=0$ then
		\item \quad $\frac{h(x)}{x-x^{-1}}\to h(x)$;
		\item end while;
		\item {\bfseries Output} $h(x)$. 
	\end{enumerate}
\end{Alg}

\begin{example}
	Let $q=2^k$, $h_1(a)=a^{\frac{1}{2}}+1$ and $h_2(a)=a+a^\frac{1}{2}+1$. Then the original 
	\begin{eqnarray*}
		h(x) &=& \left(a^{\frac{1}{2}}+1\right)x+a+a^{\frac{1}{2}}+1\\
		&=& \left(1+x^{\frac{1}{2}}+x^{-\frac{1}{2}}\right)x+x+x^{-1}+x^{\frac{1}{2}}+x^{-\frac{1}{2}}+1\\
		&=&x^{\frac{3}{2}}+x^{-1}+x^{-\frac{1}{2}}+1\\
		&=&\left(x^{\frac{1}{2}}+x^{-\frac{1}{2}}\right)\left(x^{-\frac{1}{2}}+x+1\right).
	\end{eqnarray*}
	Therefore, we obtain $h(x)=x^{-\frac{1}{2}}+x+1$ as the final output according to Algorithm \ref{alg}.
\end{example}

\begin{Rem}
	In fact, given $R(a)$ which permutes $S$  in Theorem \ref{Maintheorem}, we can obtain $h_1(a), h_2(a)$ first and then  $h(x)$ and $f(x)$ as well.  We summarize it in the following process of construction, 
	\begin{equation*}
     R(a) \rightarrow h_1(a), h_2(a) \rightarrow h(x) \rightarrow f(x).
	\end{equation*}
	
	For distinct $r>0$,  $``h(x),f(x)"$s obtained from the same $R(a)$ are generally different. However,  we claim that the corresponding fractional polynomials $``g(x)"$s are the same due to the following commutative diagram
	\begin{equation*}
	\xymatrix{
		\mu_{q+1} \ar[rr]^{g(x)}\ar[d]_{x+x^q} &   & \mu_{q+1}  \ar[d]^{x+x^q} \\
		\{2,-2\}\cup S	 \ar[rr]^{R(a)} &  & \{2,-2\} \cup S. }
	\end{equation*}
	On the other hand,   permutation polynomials over $\gf_{q^2}$ are equivalent if their corresponding fractional polynomials are  same. The unique distinction is the condition of $r$, which satisfies $(r,q-1)=1$.  Therefore, we can only consider the case $r=1$ when we construct permutation polynomials over $\gf_{q^2}$ from $R(a)$ which permutes $S$.
\end{Rem}

\begin{Cor}
	\label{Maintheorem_r=1}
	Let $f(x)=xh\left(x^{q-1}\right)\in\gf_q[x]$ and $S$ be the set defined in (\ref{S}). By Algorithm \ref{alg1}, we get $h_1(a)$ and $h_2(a)$ from $h(x)=h_1(a)x+h_2(a)$. Assume that $$R(a)=\frac{ah_1^2(a)+ah_2^2(a)+4h_1(a)h_2(a)}{h_1^2(a)+h_2^2(a)+ah_1(a)h_2(a)}.$$ Then $f(x)$ permutes $\gf_{q^2}$ if and only if the following conditions hold:
	\begin{enumerate}[(i)]
		\item for the corresponding fractional polynomial $g(x)=xh(x)^{q-1}$, $g(x)=1$ has only one solution $x=1$ in $\mu_{q+1}$ and $g(x)=-1$ has only one solution $x=-1$ in $\mu_{q+1}$;
		\item $h(x)\neq0$ for any $x\in\mu_{q+1}$;
		\item $R(a)$ permutes $\{2,-2\}\cup S$.
	\end{enumerate}
\end{Cor}

 However, it may be difficult to consider the permutation property of $R(a)$ over $S$ in Corollary \ref{Maintheorem_r=1} in general. Hence, in this paper, we mainly focus on constructing permutation polynomials of the form  $xh\left(x^{q-1}\right)$ over $\gf_{q^2}$ from some permutations of special forms over $S$.

\section{The case $\mathrm{char}\gf_{q^2}$ is even}

In the present section, we consider the case where $\mathrm{char}\gf_{q^2}$ is even. In this case, $S:=\{a\in\gf_q^{*}: \tr\left(\frac{1}{a}\right)=1\}$. In Corollary \ref{Maintheorem_r=1},  we know that 
$$\frac{1}{R(a)}=\frac{1}{a}+\frac{h_1(a)}{h_1(a)+h_2(a)}+\left(\frac{h_1(a)}{h_1(a)+h_2(a)}\right)^2,$$
and thus $R(a)$ is a mapping from $S$ to $S$ as long as  $h_1(a)+h_2(a)\neq0$ for any $a\in S$.

In the following, we obtain the even characteristic case of Corollary \ref{Maintheorem_r=1}. Then we construct some permutation polynomials over $\gf_{q^2}$ from various permutations of simple forms, which mainly are  monomials and linearized polynomials over $S$. Some of our results can explain a few known theorems.
 
 \begin{Th}
 	\label{MainTheorem_even_1}
 	Let $q=2^k,$ $S:=\{a\in\gf_q^{*}: \tr\left(\frac{1}{a}\right)=1\}$,  $f(x)=xh\left(x^{q-1}\right)\in\gf_{q}[x]$ and $h(x)=h_1(a)x+h_2(a)$ for $x\in\mu_{q+1}\backslash\{1\}$ according to Algorithm \ref{alg1}. Then $f(x)$ is a permutation polynomial over $\gf_{q^2}$ if and only if 
 	\begin{enumerate}[(1)]
        \item $h_1\left(a\right)\neq h_2(a)$ for any $a\in S$;
        \item $h(1)\neq0$;
 		\item $R(a)=\frac{1}{\frac{1}{a}+\psi(a)+\psi(a)^2}$ permutes $S$, where $\psi(a)=\frac{h_1(a)}{h_1(a)+h_2(a)}$.
 	\end{enumerate} 
 \end{Th}

\begin{proof}
	We prove the result through verifying the conditions in Corollary \ref{Maintheorem_r=1}. 

\emph{(i)} In the case, for $x\in\mu_{q+1}$, the corresponding fractional polynomial $$g(x)= xh(x)^{q-1} = \frac{h_1(a)+h_2(a)x}{h_1(a)x+h_2(a)}.$$
$g(x)=1$ if and only if $\left(h_1(a)+h_2(a)\right)(x+1)=0$. Hence, $g(x)=1$ has only one solution $x=1$ if and only if $h_1(a)\neq h_2(a)$ for any $a\in S$. Therefore, the condition (i) in Corollary \ref{Maintheorem_r=1} is equivalent to that  $h_1(a)\neq h_2(a)$ for any $a\in S$.

\emph{(ii)} If there exists  $x_0\in\mu_{q+1}$ such that $h\left(x_0\right)=0$. Let $a_0=x_0+x_0^{-1}$. Then $h\left(x_0\right)=h_1\left(a_0\right)x_0+h_2\left(a_0\right)=0$. Computing $h\left(x_0\right)+x_0h\left(x_0\right)^q$, we have $$\left(h_1\left(a_0\right)+h_2\left(a_0\right)\right)\left(x_0+1\right)=0.$$ Then $h_1\left(a_0\right)+h_2\left(a_0\right)=0$ or $x_0=1$. Since $h_1(a)\neq h_2(a)$ for any $a\in S$, we have $x_0=1$. Therefore, the second condition is equivalent to that of $h(1)\neq0$.  

\emph{(iii)} is obvious because $q$ is even.
\end{proof}

 As for Condition {(3)} in the above theorem, let $b=\frac{1}{a}$, $l(a)=\psi\left(\frac{1}{a}\right)$  and $T:=\{b\in\gf_{q}: \tr(b)=1\}$. Then it is also clear that {(3)} is equivalent to that $L(b)=b+l(b)+l(b)^2$ permutes $T$. Therefore, we also have the following theorem.
	
\begin{Th}
	\label{MainTheorem_even_2}
	 	Let $q=2^k,$ $T:=\{b\in\gf_q: \tr\left(b\right)=1\}$,  $f(x)=xh\left(x^{q-1}\right)\in\gf_{q}[x]$  and $h(x)=h_1\left(\frac{1}{b}\right)x+h_2\left(\frac{1}{b}\right)$ for $x\in\mu_{q+1}\backslash\{1\}$ according to Algorithm \ref{alg1}. Then $f(x)$ is a permutation polynomial over $\gf_{q^2}$ if and only if 
	\begin{enumerate}[(1)]
        \item $h_1\left(\frac{1}{b}\right)\neq h_2\left(\frac{1}{b}\right)$ for any $b\in T$;
        \item $h(1)\neq0$;
		\item $L(b)=b+l(b)+l(b)^2$ permutes $T$, where $l(b)=\frac{h_1\left(\frac{1}{b}\right)}{h_1\left(\frac{1}{b}\right)+h_2\left(\frac{1}{b}\right)}$. 
	\end{enumerate} 
\end{Th}	


\begin{Rem}
  \label{Rep}
   For $\psi(a)=\frac{h_1(a)}{h_1(a)+h_2(a)}=\frac{t_1(a)}{t_2(a)}$, we can see that there exists $\tau(a)\in\gf_q[a]$ such that $h_1(a)=\tau(a)t_1(a)$ and $h_2(a)=\tau(a)\left(t_1(a)+t_2(a)\right).$ Moreover, $$g(x)=\frac{h_1(a)x+h_2(a)}{h_2(a)x+h_1(a)}=\frac{t_1(a)x+t_1(a)+t_2(a)}{\left(t_1(a)+t_2(a)\right)x+t_1(a)},$$
	which means that no matter what $\tau(a)$ is, $g(x)$ is the same. Therefore, we may assume that $\tau(a)=1$, i.e., $h_1(a)=t_1(a)$ and $h_2(a)=t_1(a)+t_2(a)$. 
\end{Rem}

Given $\psi(a)=\frac{t_1(a)}{t_2(a)}$, according to Remark \ref{Rep}, without loss of generality, we assume $h_1(a)=t_1(a)$ and $h_2(a)=t_1(a)+t_2(a)$. Furthermore, we obtain $h(x)$ by Algorithm \ref{alg} and we can construct some permutation polynomials of the form  $xh\left(x^{q-1}\right)$ over $\gf_{q^2}$ from Theorem \ref{MainTheorem_even_1}. As for Theorem \ref{MainTheorem_even_2}, let $\psi(a)=l\left(\frac{1}{a}\right)$. Then we can also construct permutation polynomials by the above method.

In the following, we apply Theorem \ref{MainTheorem_even_1} or  \ref{MainTheorem_even_2} to specific cases, where $\psi(a)$ and $l(b)$ are monomials or linearized polynomials, obtaining some known permutation polynomials and new ones over $\gf_{q^2}$. Moreover, it seems that some of them cannot be proved easily by the previous approaches.

\subsection{The case of monomials}

In this subsection, we mainly consider the case where $l(b)=b^s$ and $s$ is not the power of $2$ such that $L(b)$ permutes $T$ in Theorem \ref{MainTheorem_even_2}. Using Magma, we obtain some experimental results under the conditions where $k$ is from $3$ to $12$ and $s$ is not the power of $2$, see TABLE I.  
Other than two  sporadic cases (in italic), we generalize all of the remaining cases into infinite classes.

\begin{table}[htp]
	\label{table1}
	\centering
	\caption{The value of $s\neq2^i$  such that $L(b)$ permutes $T$}
	\begin{tabular}{ c c c c c c}		
		\hline 
		\hline
		& $k$ &  $s$  & $\quad$  & $k$ &  $s$  \\
		\hline
		& $3$ & ${ 5,6}$ & $\quad$ & $4$ & ${6},{9},{ 13,14}$ \\
		\hline
		& $5$  &  $\emph{14},{21,29,30}$    &   $\quad$   &   $6$   &  $\emph{5},{28,61,62}$    \\
		\hline
		& $7$ & ${85,125,126}$    &  $\quad$  &  $8$   & ${ 120}, {165}, { 253, 254}$  \\
		\hline
		&  $9$  &  ${ 341, 509, 510}$   &   $\quad$  &   $10$   & ${496, 1021, 1022}$  \\
		\hline 
		&  $11$  &  ${ 1365, 2045, 2046}$   &  $\quad$  &  $12$  & ${2016}, {2709}, { 4093, 4094}$ \\
		\hline
		\hline
	\end{tabular}
\end{table}	

 In the following, we obtain {five}  classes of permutations over $T$ in TABLE I and thus construct several classes of permutation polynomials over $\gf_{q^2}$ by our method. Before introducing the main results in the subsection, some lemmas are needed.

\begin{Lemma}
	\cite{SZM}
	\label{lem3}
	Let $u, v \in \gf_{q}$, where $q=2^k$ and $v\neq0$. Then the cubic equation $x^3+ux+v=0$ has a unique solution in $ \gf_{q}$ if and only if $\tr\left(\frac{u^3}{v^2}+1\right) \neq 0$.
\end{Lemma}

Philip A. Leonard and Kenneth S. Williams  characterized the factorization of a quartic polynomial over $\mathbb{F}_{2^k}$ in \cite{LW}.	
\begin{Lemma}
	\cite{LW}
	\label{lemm4}
	Let $q=2^k$ and $f(x)=x^4+a_2x^2+a_1x+a_0\in\gf_q[x]$, where $a_0a_1\neq0$. Let $g(y)=y^3+a_2y+a_1$. Then  $f(x)$ is irreducible if and only if $g(y)$ only has  one solution $y_0$ in $\gf_{q}$ and $\tr\left(\frac{a_0y_0^2}{a_1^2}\right)=1$. 
\end{Lemma}

Firstly, we prove the case $s=-1$, where $L(b)$ permutes $T$. 

\begin{Lemma}
	\label{Case_Mon_s=-1}
	Let $L(b)=b+\frac{1}{b}+\frac{1}{b^2}$ and $T:=\{b\in\gf_q: \tr(b)=1\}$, where $q=2^k$. Then $L(b)$ permutes $T$.   
\end{Lemma}

\begin{proof}
	Assume that there exist $b_0, b_0+\Delta\in T$ such that $L\left(b_0\right)=L\left(b_0+\Delta\right)$, where $\Delta\in\gf_q^{*}$ and $\tr(\Delta)=0$. Then
	\begin{equation}
	\label{CASE_Mon_Th_1}
	\frac{1}{b_0}+\frac{1}{b_0^2}+\frac{1}{b_0+\Delta}+\frac{1}{b_0^2+\Delta^2}+\Delta=0.
	\end{equation}
	After simplifying Eq. (\ref{CASE_Mon_Th_1}), we obtain 
	\begin{equation}
	\label{Mon_Th_1_key}
	\Delta^2+u\Delta+v=0,
	\end{equation}
	where $u=\frac{b_0+1}{b_0^2}$ and $v=b_0^2+1$.
	Since $$\tr\left(\frac{v}{u^2}\right)=\tr\left(\frac{\left(b_0^2+1\right)b_0^4}{b_0^2+1}\right)=\tr\left(b_0^4\right)=1,$$
	according to Lemma \ref{lem2}, we have that Eq. (\ref{Mon_Th_1_key}) has no solution in $\gf_q$, which is a contradiction. Therefore, $L(b)$ permutes $T$.
\end{proof}

From Lemma \ref{Case_Mon_s=-1}, we can construct the corresponding permutation polynomial, which is over $\gf_{q^2}$. In Lemma \ref{Case_Mon_s=-1}, $l(b)=\frac{1}{b}$. Then $$\frac{h_1\left(\frac{1}{b}\right)}{h_1\left(\frac{1}{b}\right)+h_2\left(\frac{1}{b}\right)}=\frac{1}{b},$$
i.e.,
$$\frac{h_1(a)}{h_1(a)+h_2(a)}=a.$$
Let $h_1(a)=a$, $h_2(a)=a+1$ and
\begin{eqnarray*}
h(x)&=& h_1(a)x+h_2(a)\\
&=& ax+a+1\\
&=& \left(x+x^{-1}\right)x+x+x^{-1}+1\\
&=& x^2+x+x^{-1}.
\end{eqnarray*}
It is clear that $h(x)$ satisfies the conditions $\emph{(1)}$ and $\emph{(2)}$ in Theorem \ref{MainTheorem_even_2}. Therefore, $f(x)=xh\left(x^{q-1}\right)=x\left(x^{2q-2}+x^{q-1}+x^{1-q}\right)=x^{2q-1}+x^{q}+x^{q^2-q+1}$ is a permutation polynomial over $\gf_{q^2}$. {In this case, the corresponding fractional polynomial of $f(x)$ is $$g(x)=xh(x)^{q-1}=\frac{x^3+x^2+1}{x^3+x+1}.$$ In fact, in \cite{DQ}, Ding et al. proved that $f(x)=x+x^{2q-1}+x^{q^2-q+1}$,  whose corresponding fractional polynomial is also $\frac{x^3+x^2+1}{x^3+x+1}$, is a permutation over $\gf_{q^2}$, where $q=2^k$, by a direct method and some skills. Moreover, in \cite{Zhazhengbang}, Zha et al. proved that the fractional polynomial $g(x)=\frac{x^3+x^2+1}{x^3+x+1}$ permutes $\mu_{q+1}$ through showing $g(x)=\gamma$ has at most one solution in $\mu_{q+1}$ for any $\gamma\in\mu_{q+1}$ directly. Furthermore, they obtained two permutation trinomials (i.e., $x^3+x^{2q+1}+x^{3q}$ and $x^3+x^{q+2}+x^{3q}$) over $\gf_{q^2}$. Therefore, by Lemma \ref{Case_Mon_s=-1}, we can obtain these known permutation trinomials in \cite{Zhazhengbang,DQ}. By the way,  Our proof seems to be easier than theirs.  }

Next, we obtain another result ($s=-2$ in Table I). Before proving it, we give the following lemma.

\begin{Lemma}
	\label{P(x)=0nosolution}
	Let $q=2^k,$ $T:=\{b\in\gf_q: \tr(b)=1\}$ and $P(x)=x^4+x^3+b^2x^2+b^2x+b^5$, where $b\in T$. Then $P(x)=0$ has no solution in $\gf_q$. 
\end{Lemma}

\begin{proof}
	First of all, we have
	$$x^4P\left(\frac{1}{x}+b\right)=x^4\left(\frac{1}{x^4}+\frac{1}{x^3}+\frac{b^2+b}{x^2}+b^5\right)=b^5x^4+\left(b^2+b\right)x^2+x+1.$$
	Let $$P_1(x)=\frac{1}{b^5}x^4P\left(\frac{1}{x}+b\right)=x^4+a_2x^2+a_1x+a_0,$$
	where $$a_2=\frac{b+1}{b^4}, a_1=a_0=\frac{1}{b^5}.$$ Then $P(x)$ has the  same number of solutions in $\gf_q$ as $P_1(x)$.
	Next, we consider the equation 
	\begin{equation}
	\label{Cubic}
	y^3+a_2y+a_1=0.
	\end{equation}
	It is clear that $y_0=\frac{1}{b^2}$ is a solution of the above equation in $\gf_q$. Moreover, since 
	\begin{eqnarray*}
	\tr\left(\frac{a_2^3}{a_1^2}+1\right)&=& \tr\left(\frac{(b+1)^3}{b^2}+1\right)\\
	&=& \tr\left(b+1+\frac{1}{b}+\frac{1}{b^2}+1\right)\\
	&=&\tr(b)=1,
	\end{eqnarray*}
    we get that Eq. (\ref{Cubic}) has only one solution $y_0=\frac{1}{b^2}$ in $\gf_q$ from Lemma \ref{lem3}. Furthermore, we have $$\tr\left(\frac{a_0y_0^2}{a_1^2}\right)=\tr\left(b\right)=1.$$ Therefore, according to Lemma \ref{lemm4}, we know that $P(x)$ has no solution in $\gf_q$.
\end{proof}

\begin{Lemma}
	\label{s_-2}
	Let $L(b)=b+\frac{1}{b^2}+\frac{1}{b^4}$ and $T:=\{b\in\gf_q: \tr(b)=1\}$, where $q=2^k$. Then $L(b)$ permutes $T$.   
\end{Lemma}

\begin{proof}
	Let $R(a)=\frac{1}{\frac{1}{a}+a^2+a^4}$ and $S:=\{a\in\gf_q^{*}: \tr\left(\frac{1}{a}\right)=1\}$. Then $L(b)$ permutes $T$ if and only if $R(a)$ permutes $S$. In the following, we prove that $R(a)$ permutes $S$. 
	
	Assume that there exist $a, a+\Delta\in S$ such that $R(a)=R(a+\Delta)$, where $\Delta\in\gf_q^{*}$. Then 
	\begin{equation*}
	\frac{1}{a}+\frac{1}{a+\Delta}+\Delta^2+\Delta^4=0,
	\end{equation*}
	i.e., 
	\begin{equation}
	\label{mainequation}
	\Delta^4+a\Delta^3+\Delta^2+a\Delta+\frac{1}{a}=0.
	\end{equation}
	Let $P(x)=x^4+x^3+\frac{1}{a^2}x^2+\frac{1}{a^2}x+\frac{1}{a^5}$. Then $$a^4P\left(\frac{\Delta}{a}\right)=\Delta^4+a\Delta^3+\Delta^2+a\Delta+\frac{1}{a}.$$
	Therefore, Eq. (\ref{mainequation}) is equivalent to $P(x)=0$. According to Lemma \ref{P(x)=0nosolution}, we can claim that Eq. (\ref{mainequation}) is impossible. Hence, $R(a)$ permutes $S$.
\end{proof}

In the above theorem, $l(b)=\frac{1}{b^2}.$ Then $h_1(a)=a^2$ and $h_2(a)=a^2+1$. Moreover, 
\begin{eqnarray*}
h(x) &=& h_1(a)x+h_2(a) \\
&=& a^2x+a^2+1\\
&=& \left(x^2+x^{-2}\right)x+x^2+x^{-2}+1\\
&=& x^3+x^2+x^{-1}+x^{-2}+1.
\end{eqnarray*}
 Therefore, $f(x)=xh\left(x^{q-1}\right)=x\left(x^{3q-3}+x^{2q-2}+x^{1-q}+x^{2-2q}+1\right)=x^{3q-2}+x^{2q-1}+x^{q^2-q+1}+x^{q^2-2q+2}+x$ is a permutation polynomial over $\gf_{q^2}$. Hence, we have the following theorem.

\begin{Th}
	\label{PP_s=-2}
	Let $q=2^k$ and $f(x)=x^{3q-2}+x^{2q-1}+x^{q^2-q+1}+x^{q^2-2q+2}+x$. Then $f(x)$ is a permutation polynomial over $\gf_{q^2}$.
\end{Th}

In the above theorem, the corresponding fractional polynomial is  $$g(x)=xh(x)^{q-1}=\frac{x^5+x^4+x^3+x+1}{x^5+x^4+x^2+x+1}.$$
Recently, Li and Helleseth \cite{LiNian1}  proved the following permutation trinomial over $\gf_{q^2}$, where $q=2^k$ and $k$ is even:
$$f(x)=x+x^{\frac{(q-1)^2}{3}+1}+x^{\frac{(2q+7)(q-1)}{3}+1}$$
through proving the corresponding fractional polynomial  $\frac{x^7+x^5+1}{x^7+x^2+1}$, i.e., $\frac{x^5+x^4+x^3+x+1}{x^5+x^4+x^2+x+1}$ when $k$ is even, permutes $\mu_{q+1}$.

\begin{Lemma}
	\label{s3}
	Let $k$ be odd, $l(b)=b^{\frac{2^{k+1}-1}{3}}$ and $T:=\{b\in\gf_q:\tr(b)=1\}$, where $q=2^k$. Then $L(b)=b+l(b)+l(b)^2$ permutes $T$.
\end{Lemma}

\begin{proof} We claim that $L(b)$ permutes $\gf_q$ as follows. First of all, it is clear that $\gcd\left(\frac{2^{k+1}-1}{3},q-1\right)=1$ and $\left(\frac{2^{k+1}-1}{3}\right)^{-1}\equiv 3 \pmod {q-1}$. Therefore, $x^3$ permutes $\gf_q$ and it suffices to show that $L\left(b^3\right)=b^3+b+b^2=(b+1)^3+1$ permutes $\gf_q$, which is obvious. Because $\tr(L(b)) =\tr(b)$, $L(b)$ permutes $T$ as well.
\end{proof}

In Lemma \ref{s3}, $l(b)=b^s$, where $s=\frac{2^{k+1}-1}{3}.$ Then $\psi(a)=\frac{1}{a^s}$ and $h_1(a)=1$, $h_2(a)=a^s+1$  according to Algorithm \ref{alg1}. Moreover, the original 
\begin{eqnarray*}
h(x) &=& x+a^s+1\\
&=& x+1+\left(x+x^{-1}\right)^s\\
&=& \left(x^{\frac{1}{2}}+x^{-\frac{1}{2}}\right)\left(x^{\frac{1}{2}}+\left(x+x^{-1}\right)^{s-1}\left(x^{\frac{1}{2}}+x^{-\frac{1}{2}}\right)\right).
\end{eqnarray*}
Therefore, we obtain $h(x)=x^{\frac{1}{2}}+\left(x+x^{-1}\right)^{s-1}\left(x^{\frac{1}{2}}+x^{-\frac{1}{2}}\right)$ according to Algorithm \ref{alg}. Moreover, $$f\left(x^2\right)=x^2h\left(x^{2(q-1)}\right)=x^{q+1}+x^2\left(x^{q-1}+x^{q^2-q}\right)^{2s-1}.$$ Hence we obtain the following theorem. 
\begin{Th}
	Let $q=2^k$, $k$ be odd and $s=\frac{2^{k+1}-1}{3}$. Then $f(x)=x^{q+1}+x^2\left(x^{q-1}+x^{q^2-q}\right)^{2s-1}$ is a permutation polynomial over $\gf_{q^2}$.
\end{Th}

\begin{Lemma}
	\label{s4}
	Let $k=2m$, $q=2^k$, $T:=\{b\in\gf_q: \tr(b)=1\}$ and $l(b)=b^{2^{k-1}-2^{m-1}}$. Then $L(b)=b+l(b)+l(b)^2$ permutes $T$.
\end{Lemma}

\begin{proof}
Let $q_1=2^m$. We claim that $L(b)^{2q_1}$ permutes $T$  by proving that  the equation $L(b)^{2q_1}=c$, i.e., 
\begin{equation*}
b^{2q_1}+b^{q_1-1}+b^{2q_1-2}=c,
\end{equation*} 
has at most one solution in $T$ for any $c\in T$. 

Multiplying both sides  of the above equation by $b^2$, we get 
\begin{equation}
\label{L(b)=c}
b^{2q_1+2}+b^{q_1+1}+b^{2q_1}=b^2c.
\end{equation}

Computing the sum of Eq. (\ref{L(b)=c}) and the $q_1$-th power of  Eq. (\ref{L(b)=c}), we have $b^2(c+1)=b^{2q_1}\left(c^{q_1}+1\right)$. Therefore, $b^2(c+1)\in\gf_{q_1}$. Assume $b^2(c+1)=u$, where $u\in\gf_{q_1}$.  Since $k$ is even, we have $1\not\in T$, which means $c+1\ne0$ for any $c\in T$. Hence, $$b=\left(\frac{u}{c+1}\right)^{\frac{1}{2}}, b^2=\left(\frac{u}{c+1}\right),  b^{2q_1} = \frac{u}{c^{q_1} + 1},  b^{2q_1+2} = \frac{u^2}{c^{q_1+1} + c^{q_1} + c + 1}. $$
Raising Eq. (\ref{L(b)=c}) to a power of two and plugging the above powers of $b$'s into it, we obtain 
$$u^2=c^{2q_1+2}+c^{q_1+1}+c^{q_1}+c.$$
Therefore, $$b=\left(\frac{c^{2q_1+2}+c^{q_1+1}+c^{q_1}+c}{c^2+1}\right)^{\frac{1}{4}}.$$
That means Eq. (\ref{L(b)=c}) has at most one solution in $T$.
\end{proof}

In Lemma \ref{s4}, $\psi(a)=l\left(\frac{1}{a}\right)=\frac{a^{2^{m-1}}}{a^{2^{k-1}}}$. Then we have $h_1(a)=a^{2^{m-1}}$ and $h_2(a)=a^{2^{k-1}}+a^{2^{m-1}}.$ Moreover, according to Algorithm \ref{alg}, the original 
\begin{eqnarray*}
h(x) &=& a^{2^{m-1}} x+ a^{2^{k-1}}+a^{2^{m-1}} \\
&=& \left(x+x^{-1}\right)^{2^{m-1}} \left(x+ \left(x+x^{-1}\right)^{2^{k-1}-2^{m-1}}+1\right)\\
&=& \left(x+x^{-1}\right)^{2^{m-1}+\frac{1}{2}}\left(x^{\frac{1}{2}}+\left(x+x^{-1}\right)^{2^{k-1}-2^{m-1-\frac{1}{2}}}\right),
\end{eqnarray*}
and finally we obtain $h(x)=x^{\frac{1}{2}}+\left(x+x^{-1}\right)^{2^{k-1}-2^{m-1}-\frac{1}{2}}$. Moreover, $$f\left(x^2\right)=x^2h\left(x^{2(q-1)}\right)=x^2\left(x^{q-1}+\left(x^{q-1}+x^{1-q}\right)^{2^k-2^m-1}\right).$$
Hence, we have the following theorem.

\begin{Th}
	Let $q=2^k$, where $k=2m$ and $f(x)=x^{q+1}+x^2\left(x^{q-1}+x^{1-q}\right)^{2^k-2^m-1}$. Then $f(x)$ is a permutation polynomial over $\gf_{q^2}$. 
\end{Th}

{ Finally, we give the last result in this subsection.

\begin{Lemma}
	Let $q=2^k$, where $k$ is even, $T:= \{b\in\gf_{q^2}^{*} : \tr(b)=1 \}$ and $s={\frac{2q^2-q-1}{3}}$. Then $L(b)=b+b^s+b^{2s}$ permutes $T$.
\end{Lemma}

\begin{proof}
	Let $d=b^{\frac{q-1}{3}}$, $w=b^{\frac{q^2-1}{3}}$. Then  $w^3=1$, $w\in\gf_q$, $d^{q+1}=w$, $d^3=b^{q-1}$ and  $b^s=\frac{w^2}{d}$. For any $c\in T$, we claim that $L(b)=c$, i.e., 
	\begin{equation}
	\label{b+wd=c}
	b+\frac{1}{wd}+\frac{1}{w^2d^2}=c,
	\end{equation}
	has at most one solution in $T$. 
	
	Raising Eq. (\ref{b+wd=c}) into its $q$-th power, we obtain 
	\begin{equation*}
	b^q+wd+w^2d^2=c^q,
	\end{equation*}
	i.e., 
	\begin{equation}
	\label{b^q+wd=c^q}
	d^3b+wd+w^2d^2=c^q.
	\end{equation}
	Computing $(\ref{b+wd=c})*d^3+(\ref{b^q+wd=c^q})$, we get $d^3=c^{q-1}$. Hence, $b=\gamma c$, where $\gamma\in\gf_q^{*}$. In fact, $$\gamma^s=\gamma^{\frac{q-1}{3}(2q+1)}=\gamma^{q-1}=1.$$ Plugging $b=\gamma c$ into $L(b)=c$, we have $b+\gamma^sc^s+\gamma^{2s}c^{2s}=c$. Therefore, $$b=c+c^s+c^{2s} \in T.$$
\end{proof}
Using the above lemma, we can obtain the following class of permutation polynomials similarly.

\begin{Th}
	Let $q=2^k$, $k$ be even and $s={\frac{2q^2-q-1}{3}}$. Then $f(x)=x^{q^2+1}+x^2\left(x^{q^2-1}+x^{q^4-q^2}\right)^{2s-1}$ is a permutation polynomial over $\gf_{q^4}$.
\end{Th}
}


\subsection{The case of linearized polynomials}
\label{The case of linearized polynomial}
In this subsection, we consider the case where $l(b)$ is a linearized polynomial in Theorem \ref{MainTheorem_even_2}. For a linearized polynomial $L(b)=b+l(b)+l(b)^2$ over $\gf_q$, it is known that $L(b)$ permutes $\gf_q$ if and only if $L(b)=0$ has only zero solution in $\gf_q$. In the following, we prove the fact that $L(b)$ permutes $T$  is equivalent to that $L(b)$ permutes $\gf_q$, where $T:=\{b\in\gf_q: \tr(b)=1\}$.

\begin{Lemma}
	\label{LpermutesT}
	Let $q=2^k$, $T:=\{b\in\gf_q: \tr(b)=1\}$ and $l(b)$ be a linearized polynomial over $\gf_q$. Then $L(b)=b+l(b)+l(b)^2$ is also linearized  over $\gf_q$. Then $L(b)$ permutes $T$ if and only if $L(b)$ is a permutation polynomial over $\gf_q$.  
\end{Lemma}

\begin{proof}
	(1) If $L(b)$ is a permutation polynomial over $\gf_q$, it is obvious that $L(b)$ permutes $T$ because $\tr(L(b)) = \tr(b)$.
	
	(2) If $L(b)$ permutes $T$ and $L(b)$ is not a permutation polynomial over $\gf_q$, there exists $b_0\in\gf_q^{*}$ such that $L\left(b_0\right)=0$ and $\tr\left(b_0\right)=0.$ Then for any $b_1\in T$, $b_0+b_1\in T$, $L\left(b_0+b_1\right)=L\left(b_1\right),$ which means that $L(b)$ does not permute $T$, a contradiction.   Therefore, $L(b)$ is a permutation polynomial over $\gf_q$.
\end{proof}

In the following, we construct permutation polynomials $f(x)=xh\left(x^{q-1}\right)$ over $\gf_{q^2}$ from linearized permutation polynomials $L(b)=b+l(b)+l(b)^2$ over $\gf_q$. Let $l(b)=\sum_{i=0}^{t}\alpha_ib^{2^i}$, where $t\le k-1$, $\alpha_i\in\gf_q$ and $\alpha_t\neq0$. Assume that $L(b)=b+l(b)+l(b)^2$ is a permutation polynomial over $\gf_q$, i.e., over $T$ according to Lemma \ref{LpermutesT}. Then $$\psi(a)=l\left(\frac{1}{a}\right)=\frac{\sum_{i=0}^{t}\alpha_ia^{2^t-2^i}}{a^{2^t}}=\frac{h_1(a)}{h_1(a)+h_2(a)}.$$

Therefore, $h_1(a)=\sum_{i=0}^{t}\alpha_ia^{2^t-2^i}$, $h_2(a)=\sum_{i=0}^{t}\alpha_ia^{2^t-2^i}+a^{2^t}$. Moreover, the original
\begin{eqnarray*}
h(x) &=& h_1(a)x+h_2(a) \\
&=& \left(\sum_{i=0}^{t-1}\alpha_ia^{2^t-2^i}\right)x+\alpha_tx+\sum_{i=0}^{t-1}\alpha_ia^{2^t-2^i}+\alpha_t+a^{2^t}\\
&=&\left(x^{\frac{1}{2}}+x^{-\frac{1}{2}}\right)\left(\left(\sum_{i=0}^{t-1}\alpha_i\left(x+x^{-1}\right)^{2^t-2^i-\frac{1}{2}}\right)x+\sum_{i=0}^{t-1}\alpha_i\left(x+x^{-1}\right)^{2^t-2^i-\frac{1}{2}}+\left(x+x^{-1}\right)^{2^t-\frac{1}{2}}+\alpha_tx^{\frac{1}{2}}\right),
\end{eqnarray*}
and we obtain $h(x)=\left(\sum_{i=0}^{t-1}\alpha_i\left(x+x^{-1}\right)^{2^t-2^i-\frac{1}{2}}\right)x+\sum_{i=0}^{t-1}\alpha_i\left(x+x^{-1}\right)^{2^t-2^i-\frac{1}{2}}+\left(x+x^{-1}\right)^{2^t-\frac{1}{2}}+\alpha_tx^{\frac{1}{2}}$ according to Algorithm \ref{alg}. We raise $f(x)$ to its square for ease of presentation. Then $$f\left(x^2\right)=x^2h\left(x^{2(q-1)}\right)=x^2\left(e(x)x^{2(q-1)}+e(x)+\left(x^{q-1}+x^{q^2-q}\right)^{2^{t+1}-1}+\alpha_tx^{q-1}\right),$$  where $e(x)=\sum_{i=0}^{t-1}\alpha_i\left(x^{q-1}+x^{q^2-q}\right)^{2^{t+1}-2^{i+1}-1}.$

\begin{Th}
	\label{linear_maintheorem}
	Let $q=2^k,$ $l(b)=\sum_{i=0}^{t}\alpha_ib^{2^i}$, where $t \leq k-1$, $\alpha_i \in \gf_q$, and $\alpha_t \neq0$.  If  $L(b)=b+l(b)+l(b)^2$ permute $\gf_q$, then $$f(x)=x^2\left(e(x)x^{2(q-1)}+e(x)+\left(x^{q-1}+x^{q^2-q}\right)^{2^{t+1}-1}+\alpha_tx^{q-1}\right),$$  where $e(x)=\sum_{i=0}^{t-1}\alpha_i\left(x^{q-1}+x^{q^2-q}\right)^{2^{t+1}-2^{i+1}-1},$ is a permutation polynomial over $\gf_{q^2}$.
\end{Th}

 In the following, we obtain some explicit permutation polynomials over $\gf_{q^2}$ from specific permutations over $T$. The corresponding fractional polynomials have complicated formats and  high-degree. Hence, it is normally  difficult to obtain them by the fractional approach. However, we can prove them using this approach easily.

{\bfseries \emph{B.1  $L(b)$ is a linearized monomial}}  

Let $l(b)=\sum_{i=0}^{t}\alpha_ib^{2^i}$, where $t\le k-1$. In this  subsection, we classify all cases where $L(b)=b+l(b)+l(b)^2$ is a monomial.

\begin{Lemma}
	\label{L(b)_is_mon}
	Let $q=2^k$ and $l(b)=\sum_{i=0}^{t}\alpha_ib^{2^i}$ over $\gf_{q}$, where $t\le k-1$. Then $L(b)=b+l(b)+l(b)^2$ is a monomial if and only if $l(b)$ is one of the following cases:
	\begin{enumerate}[(i)]
		\item $t=k-1$ and $\left(\alpha_0,\alpha_1,\cdots,\alpha_{k-1}\right)=\left(\alpha_0,\alpha_0^2,\cdots,\alpha_{0}^{2^{k-1}}\right)$, where $\alpha_{0}\in\gf_q^{*}$;
		\item $t=k-1$ and $\left(\alpha_0,\alpha_1,\cdots,\alpha_j,\cdots,\alpha_{k-1}\right)=\left(1+\alpha_j^{2^{k-j}},1+\alpha_j^{2^{k+1-j}},\cdots,\alpha_j,\cdots, \alpha_j^{2^{k-j-1}}\right),$ where $j\in[1,k-1]$ and $\alpha_j\in\gf_q^{*}$;
		\item $t\le k-2$ and $\left(\alpha_0,\alpha_1,\cdots,\alpha_{t}\right)=\left(1,1,\cdots,1\right)$.
	\end{enumerate}
\end{Lemma}

\begin{proof}
	From $l(b)=\sum_{i=0}^{t}\alpha_ib^{2^i}$, we have 
	\begin{eqnarray*}
	L(b) &=& b+\sum_{i=0}^{t}\alpha_ib^{2^i}+\sum_{i=0}^{t}\alpha_i^2b^{2^{i+1}}  \\
	&=& \left(1+\alpha_0\right)b+\sum_{i=1}^{t}\left(\alpha_i+\alpha_{i-1}^2\right)b^{2^i}+\alpha_t^2b^{2^{t+1}}.
	\end{eqnarray*}
We divide the proof into two cases. 

\emph{(1)} If $t={k-1}$, then $b^{2^{t+1}}=b$ and $L(b)=\left(1+\alpha_0+\alpha_t^2\right)b+\sum_{i=1}^{k-1}\left(\alpha_i+\alpha_{i-1}^2\right)b^{2^i}.$ Since $L(b)$ is a monomial, we have (i) $1+\alpha_0+\alpha_{k-1}^2\neq0$ and $\alpha_i+\alpha_{i-1}^2=0$ for all $i\in [1,k-1]$; or (ii) there exists one $j\in [1,k-1]$ such that $\alpha_j+\alpha_{j-1}^2\neq0$, $1+\alpha_0+\alpha_t^2=0$ and $\alpha_i+\alpha_{i-1}^2=0$ for all $i\in [1,k-1]\backslash\{j\}$. As for the case (i), we have $\alpha_i=\alpha_{i-1}^2$ for all $i\in [1,k-1]$. Moreover, $\alpha_i=\alpha_0^{2^i}$ for all $i\in [1,k-1]$.  And $1+\alpha_0+\alpha_{k-1}^2=1\neq0$. Therefore, in this case, we have $$\left(\alpha_0,\alpha_1,\cdots,\alpha_{k-1}\right)=\left(\alpha_0,\alpha_0^2,\cdots,\alpha_{0}^{2^{k-1}}\right),$$
where $\alpha_0\in\gf_q^{*}$. As for the case (ii), firstly, we have $\alpha_i=\alpha_j^{2^{i-j}}$ for $i\in [j,k-1]$. On the other hand, $\alpha_{0}=1+\alpha_{k-1}^2=1+\alpha_j^{2^{k-j}}$. Furthermore, $\alpha_i=\alpha_{0}^{2^{i}}=1+\alpha_j^{2^{k+i-j}}$ for any $i\in[1,j-1]$. Therefore, in the case, we have 
$$\left(\alpha_0,\alpha_1,\cdots,\alpha_j,\cdots,\alpha_{k-1}\right)=\left(1+\alpha_j^{2^{k-j}},1+\alpha_j^{2^{k+1-j}},\cdots,\alpha_j,\cdots, \alpha_j^{2^{k-j-1}}\right),$$
where $j\in[1,k-1]$ and $\alpha_j\in\gf_q^{*}$.

\emph{(2)} If $t\le k-2$, then $1+\alpha_{0}=0$ and $\alpha_i=\alpha_{i-1}^2$ for all $i\in [1,t]$. Therefore, 
$$\left(\alpha_0,\alpha_1,\cdots,\alpha_{t}\right)=\left(1,1,\cdots,1\right).$$
\end{proof}

From  Lemma \ref{L(b)_is_mon}, we can construct some explicit classes of permutation polynomials over $\gf_{q^2}$.

\begin{Th}
	Let $q=2^k$. Then $$f(x)=x^2\left(e(x)x^{2(q-1)}+e(x)+\left(x^{q-1}+x^{q^2-q}\right)^{2^{t+1}-1}+\alpha_tx^{q-1}\right)\in\gf_{q}[x],$$  where $e(x)=\sum_{i=0}^{t-1}\alpha_i\left(x^{q-1}+x^{q^2-q}\right)^{2^{t+1}-2^{i+1}-1},$ is a permutation polynomial over $\gf_{q^2}$ if one of the following conditions holds:
		\begin{enumerate}[(i)]
		\item $t=k-1$ and $\left(\alpha_0,\alpha_1,\cdots,\alpha_{k-1}\right)=\left(\alpha_0,\alpha_0^2,\cdots,\alpha_{0}^{2^{k-1}}\right)$, where $\alpha_{0}\in\gf_q^{*}$;
		\item $t=k-1$ and $\left(\alpha_0,\alpha_1,\cdots,\alpha_j,\cdots,\alpha_{k-1}\right)=\left(1+\alpha_j^{2^{k-j}},1+\alpha_j^{2^{k+1-j}},\cdots,\alpha_j,\cdots, \alpha_j^{2^{k-j-1}}\right),$ where $j\in[1,k-1]$ and $\alpha_j\in\gf_q^{*}$;
		\item $t\le k-2$ and $\left(\alpha_0,\alpha_1,\cdots,\alpha_{t}\right)=\left(1,1,\cdots,1\right)$.
	\end{enumerate}
\end{Th}
\begin{proof}
 Let $l(b)=\sum_{i=0}^{t}\alpha_ib^{2^i}$. Then from Theorem \ref{L(b)_is_mon}, we know that $L(b)=b+l(b)+l(b)^2$ are all monomials under these conditions. Therefore, it is trivial that $L(b)$ permutes $\gf_q$. Furthermore, according to Theorem \ref{linear_maintheorem}, we know that $f(x)$ permutes $\gf_{q^2}$.
\end{proof}

{\bfseries \emph{B.2 $L(b)$ is a linearized binomial}}

Let  $l(b)=\sum_{i=0}^{t}\alpha_ib^{2^i}$.  We recall that 
	$$L(b)=\left(1+\alpha_0\right)b+\sum_{i=1}^{t}\left(\alpha_i+\alpha_{i-1}^2\right)b^{2^i}+\alpha_t^2b^{2^{t+1}}.$$
	In this subsection, we classify all cases where $L(b)$ is a linearized binomial.
	
\begin{Lemma}
	\label{binomial_case}
	Let $q=2^k$, $l(b)=\sum_{i=0}^{t}\alpha_ib^{2^i},$ where $\alpha_i\in\gf_{q}$ and $\alpha_t\neq0$. Then $L(b)=b+l(b)+l(b)^2$ is a binomial if and only if one of the following cases hold:
	\begin{enumerate}[(i)]
		\item  $t=k-1$ and $\left(\alpha_{0},\cdots,\alpha_{j-1},\alpha_j,\alpha_{j+1},\cdots,\alpha_{k-1}\right)=\left(\alpha_{0},\cdots,\alpha_{0}^{2^{j-1}},\alpha_j,\alpha_j^2,\cdots,\alpha_j^{2^{k-j-1}}\right),$ where $1\le j\le k-1$, $\alpha_{j}\in\gf_q^{*},\alpha_{0}\in\gf_{q}$ satisfy that $\alpha_j\neq\alpha_{0}^{2^j}$, $\alpha_j^{2^{k-j}}+\alpha_{0}+1\neq0$;
		\item $t=k-1$ and 
			\begin{eqnarray*}
			& & \left(\alpha_0,\cdots,\alpha_{j_1-1},\alpha_{j_1},\cdots,\alpha_{j_2-1},\alpha_{j_2},\cdots,\alpha_{k-1}\right)\\
			&=& \left(1+\alpha_{j_2}^{2^{k-j_2}},\cdots,1+\alpha_{j_2}^{2^{k-j_2+j_1-1}},\alpha_{j_1},\cdots,\alpha_{j_1}^{2^{j_2-j_1-1}},\alpha_{j_2},\cdots,\alpha_{j_2}^{2^{k-j_2-1}}\right)
		\end{eqnarray*}
		where  $\alpha_{j_1}+\alpha_{j_2}^{2^{k-j_2+j_1}}+1\neq0$, $\alpha_{j_2}\neq\alpha_{j_1}^{2^{j_2-j_1}}$  and $1\le j_1< j_2\le k-1$.
		\item $t\le k-2$ and $\left(\alpha_0,\alpha_1,\cdots,\alpha_{t}\right)=\left(\alpha_0,\alpha_0^2,\cdots,\alpha_{0}^{2^{t}}\right),$
		where $\alpha_0\in\gf_q\backslash\{0,1\};$
		\item $t\le k-2$ and  $\left(\alpha_{0},\alpha_1,\cdots,\alpha_{j-1},\alpha_j,\cdots,\alpha_t\right)=\left(1,1,\cdots,1,\alpha_j,\cdots,\alpha_j^{2^{t-j}}\right),$
		where $1\le j\le t$ and $\alpha_j\in\gf_q\backslash\{0,1\}$.
	\end{enumerate}
\end{Lemma}

\begin{proof}
	\emph{(1)} If $t=k-1$, then $$L(b)=\left(1+\alpha_{0}+\alpha_t^2\right)b+\sum_{i=1}^{t}\left(\alpha_i+\alpha_{i-1}^2\right)b^{2^i}.$$
	Since $L(b)$ is a binomial, we have two cases: (i) $1+\alpha_{0}+\alpha_t^2\neq0$, there exists one $j\in[1,k-1]$ such that $\alpha_j+\alpha_{j-1}^2\neq0$ and $\alpha_i+\alpha_{i-1}^2=0$ for all $i\in [1,k-1]\backslash\{j\}$; (ii) $1+\alpha_{0}+\alpha_t^2=0$, there exist two $j_1,j_2\in[1,k-1]$ such that $\alpha_{j_1}+\alpha_{j_{1}-1}^2\neq0$, $\alpha_{j_2}+\alpha_{j_{2}-1}^2\neq0$ and $\alpha_i+\alpha_{i-1}^2=0$ for all $i\in [1,k-1]\backslash\{j_1,j_2\}$. As for the first case (i), we have $$\left(\alpha_{0},\cdots,\alpha_{j-1},\alpha_j,\alpha_{j+1},\cdots,\alpha_{k-1}\right)=\left(\alpha_{0},\cdots,\alpha_{0}^{2^{j-1}},\alpha_j,\alpha_j^2,\cdots,\alpha_j^{2^{k-j-1}}\right),$$
	where $\alpha_{j}\in\gf_q^{*},\alpha_{0}\in\gf_{q}$ satisfy that $\alpha_j\neq\alpha_{0}^{2^j}$, $\alpha_j^{2^{k-j}}+\alpha_{0}+1\neq0$. As for the second case (ii), we have 
	\begin{eqnarray*}
	& & \left(\alpha_0,\cdots,\alpha_{j_1-1},\alpha_{j_1},\cdots,\alpha_{j_2-1},\alpha_{j_2},\cdots,\alpha_{k-1}\right)\\
	&=& \left(1+\alpha_{j_2}^{2^{k-j_2}},\cdots,1+\alpha_{j_2}^{2^{k-j_2+j_1-1}},\alpha_{j_1},\cdots,\alpha_{j_1}^{2^{j_2-j_1-1}},\alpha_{j_2},\cdots,\alpha_{j_2}^{2^{k-j_2-1}}\right)
	\end{eqnarray*}
    where  $\alpha_{j_1}+\alpha_{j_2}^{2^{k-j_2+j_1}}+1\neq0$, $\alpha_{j_2}\neq\alpha_{j_1}^{2^{j_2-j_1}}$.
    
    \emph{(2)} If $t\le k-2$,  then
$$L(b)=\left(1+\alpha_0\right)b+\sum_{i=1}^{t}\left(\alpha_i+\alpha_{i-1}^2\right)b^{2^i}+\alpha_t^2b^{2^{t+1}}.$$
There  are also  two cases: (i) $1+\alpha_{0}\neq0$ and $\alpha_i+\alpha_{i-1}^2=0$ for all $i\in [1,t]$;  (ii) $1+\alpha_{0}=0$, there exists one $j\in[1,t]$ such that $\alpha_j+\alpha_{j-1}^2\neq0$ and $\alpha_i+\alpha_{i-1}^2=0$ for all $i\in [1,t]\backslash\{j\}$. 
    As for the first case, we can obtain 
    $$\left(\alpha_0,\alpha_1,\cdots,\alpha_{t}\right)=\left(\alpha_0,\alpha_0^2,\cdots,\alpha_{0}^{2^{t}}\right),$$
    where $\alpha_0\in\gf_q\backslash\{0,1\}$. As for the second case, we have 
    $$\left(\alpha_{0},\alpha_1,\cdots,\alpha_{j-1},\alpha_j,\cdots,\alpha_t\right)=\left(1,1,\cdots,1,\alpha_j,\cdots,\alpha_j^{2^{t-j}}\right),$$
    where $\alpha_j\in\gf_q\backslash\{0,1\}$.
\end{proof}

The following result is well known.

\begin{Lemma}
	\cite{WB}
	\label{Linearized_binomial}
	Let $c\in\gf_{2^k}^{*}$ and let $d:=(k,r)$. Then $L_{c,r}(x)=x^{2^r}+cx\in\gf_{2^k}[x]$ permutes $\gf_{2^k}$ if and only if $\mathrm{N}_{2^k/2^d}(c)\neq1$ , where  $\mathrm{N}_{2^k/2^d}(x)=x^{\frac{2^k-1}{2^d-1}}$ is the norm function from $\gf_{2^k}$ to $\gf_{2^d}$.
\end{Lemma}
\begin{Th}
	\label{Lin_bin_maintheorem}
	Let $q=2^k$. Then $$f(x)=x^2\left(e(x)x^{2(q-1)}+e(x)+\left(x^{q-1}+x^{q^2-q}\right)^{2^{t+1}-1}+\alpha_tx^{q-1}\right)\in\gf_{q}[x],$$  where $e(x)=\sum_{i=0}^{t-1}\alpha_i\left(x^{q-1}+x^{q^2-q}\right)^{2^{t+1}-2^{i+1}-1},$ is a permutation polynomial over $\gf_{q^2}$ if one of the following conditions holds:
	\begin{enumerate}[(i)]
	\item  $t=k-1$ and $\left(\alpha_{0},\cdots,\alpha_{j-1},\alpha_j,\alpha_{j+1},\cdots,\alpha_{k-1}\right)=\left(\alpha_{0},\cdots,\alpha_{0}^{2^{j-1}},\alpha_j,\alpha_j^2,\cdots,\alpha_j^{2^{k-j-1}}\right),$ where $1\le j\le k-1$,  $d=(k,j)>1$, $\alpha_{j}\in\gf_q^{*},\alpha_{0}\in\gf_{q}$ satisfy that $\alpha_j\neq\alpha_{0}^{2^j}$, $\alpha_j^{2^{k-j}}+\alpha_{0}+1\neq0$ and $\mathrm{N}_{2^k/2^d}\left(\frac{\alpha_j^{2^{k-j}}+\alpha_{0}+1}{\alpha_j+\alpha_{0}^{2^j}}\right)\neq1$;
	\item $t=k-1$ and 
	\begin{eqnarray*}
		& & \left(\alpha_0,\cdots,\alpha_{j_1-1},\alpha_{j_1},\cdots,\alpha_{j_2-1},\alpha_{j_2},\cdots,\alpha_{k-1}\right)\\
		&=& \left(1+\alpha_{j_2}^{2^{k-j_2}},\cdots,1+\alpha_{j_2}^{2^{k-j_2+j_1-1}},\alpha_{j_1},\cdots,\alpha_{j_1}^{2^{j_2-j_1-1}},\alpha_{j_2},\cdots,\alpha_{j_2}^{2^{k-j_2-1}}\right)
	\end{eqnarray*}
	where  $\alpha_{j_1}+\alpha_{j_2}^{2^{k-j_2+j_1}}+1\neq0$, $\alpha_{j_2}\neq\alpha_{j_1}^{2^{j_2-j_1}}$  and $1\le j_1< j_2\le k-1$,  $d=\left(k-j_2+j_1,k\right)>1$ and $\mathrm{N}_{2^k/2^d}\left(\frac{\alpha_{j_2}+\alpha_{j_1}^{2^{j_2-j_1}}}{\alpha_{j_1}+\alpha_{j_2}^{2^{k-j_2+j_1}}+1}\right)\neq1$;
	\item $t\le k-2$, $d=(t+1,k)>1$ and $\left(\alpha_0,\alpha_1,\cdots,\alpha_{t}\right)=\left(\alpha_0,\alpha_0^2,\cdots,\alpha_{0}^{2^{t}}\right),$
	where $\alpha_0\in\gf_q\backslash\{0,1\}$, $\mathrm{N}_{2^k/2^d}\left(\frac{1+\alpha_{0}}{\alpha_{0}^{2^{t+1}}}\right)\neq1$;
	\item $t\le k-2$ and  $\left(\alpha_{0},\alpha_1,\cdots,\alpha_{j-1},\alpha_j,\cdots,\alpha_t\right)=\left(1,1,\cdots,1,\alpha_j,\cdots,\alpha_j^{2^{t-j}}\right),$
	where $1\le j\le t$, $d=(k+j-t-1,k)>1$ and $\alpha_j\in\gf_q\backslash\{0,1\}$ and $\mathrm{N}_{2^k/2^d}\left(\frac{\alpha_{j}^{2^{t-j+1}}}{\alpha_{j}+1}\right)\neq1$.
\end{enumerate}
\end{Th}

\begin{proof}
	The proof can be obtained by Theorem \ref{MainTheorem_even_2}, Lemma \ref{binomial_case} and Lemma \ref{Linearized_binomial}. We omit  the details here. 
\end{proof}

Actually, there are numerous classes of PPs that can be obtained from Theorem \ref{Lin_bin_maintheorem}. We demonstrate one example in the following corollary.

\begin{Cor}
	Let $t=2$, $k\ge4$, $k\equiv0\pmod3$,  $\alpha\in\gf_{q}\backslash\{0,1\}$ and $\mathrm{N}_{2^k/2^3}\left(\frac{1+\alpha}{\alpha^8}\right)\neq1$. Then $f(x)=(\alpha+1)x^{7q-5}+\left(\alpha^2+1\right)x^{5q-3}+(\alpha+1)x^{3q-1}+\left(\alpha^4+1\right)x^{q+1}+(\alpha+1)x^{q^2-q+2}+\left(\alpha^2+1\right)x^{q^2-3q+4}+(\alpha+1)x^{q^2-5q+6}+x^{q^2-7q+8}$ is a permutation polynomial over $\gf_{q^2}$.
\end{Cor}
\begin{proof}
	In the Case (iii) of Theorem \ref{Lin_bin_maintheorem}, $t=2,$ $d=(t+1,k)=3>1$, then we can obtain a class of permutation polynomial from  this theorem after simplifying the conditions in Theorem \ref{Lin_bin_maintheorem}. 
\end{proof}

{\bfseries \emph{B.3 $L(b)$ is a linearized trinomial}} 

In this subsection, we consider two cases such that $L(b)$ is a linearized trinomial. 

\begin{Lemma}
	\label{l(b)=b^4}
	Let $q=2^k$, $T:=\{b\in\gf_q:\tr(b)=1\}$ and $L(b)=b+b^4+b^8$. Then $L(b)$ permutes $T$ (or $\gf_q$)  if and only if $k\not\equiv0\pmod7$.
\end{Lemma} 

In this  lemma, $l(b)=b^4$.  The proof is clear so we omit all the details.  Therefore, we can construct a class of  permutation polynomial over $\gf_{q^2}$ as follows.

\begin{Th}
	Let $q=2^k$ and $f(x)=x^{7q-5}+x^{5q-3}+x^{3q-1}+x^{q^2-q+2}+x^{q^2-3q+4}+x^{q^2-5q+6}+x^{q^2-7q+8}$. Then $f(x)$ is a permutation polynomial over $\gf_{q^2}$ if and only if  $k\not\equiv0\pmod7$.
\end{Th}

Using the fractional approach to prove the above theorem, we have to prove the fractional polynomial $$g(x)=\frac{x^{16}+x^{14}+x^{12}+x^{10}+x^6+x^4+x^2}{x^{14}+x^{12}+x^{10}+x^{6}+x^4+x^2+1},$$
i.e.,
$$\frac{x^{8}+x^{7}+x^{6}+x^{5}+x^3+x^2+x}{x^{7}+x^{6}+x^{5}+x^{3}+x^2+x+1},$$
 permutes $\mu_{q+1}$, which  seems to be pretty difficult.

In addition, let $l(b)=b^2+b^4$. Then $L(b)=b+b^2+b^8$ and we have the following lemma, which is easy to prove.

\begin{Lemma}
	Let $q=2^k$, $T:=\{b\in\gf_q:\tr(b)=1\}$ and $L(b)=b+b^2+b^8$. Then $L(b)$ permutes $T$  (or $\gf_q$) if and only if $k\not\equiv0\pmod7$.
\end{Lemma}

According to the above lemma, we can construct a class of permutation polynomial over $\gf_{q^2}$.

\begin{Th}
	Let $q=2^k$  and $f(x)=x^{7q-5}+x^{3q-1}+x^{q^2-q+2}+x^{q^2-5q+6}+x^{q^2-7q+8}$. Then $f(x)$ is a permutation polynomial over $\gf_{q^2}$ if and only if $k\not\equiv0\pmod7$.
\end{Th}

In the above theorem, using the fractional approach, we should prove the fractional polynomial $$g(x)=\frac{x^{16}+x^{14}+x^{10}+x^6+x^2}{x^{14}+x^{10}+x^{6}+x^2+1},$$
i.e.,
$$\frac{x^{8}+x^{7}+x^{5}+x^3+x}{x^{7}+x^{5}+x^{3}+x+1}$$
permutes $\mu_{q+1}$, 
which seems to be difficult, too. However, using our method, it can be proved easily.

\section{The case $\mathrm{char}\gf_{q^2}$ is odd}

In this section, we consider the case where $\mathrm{char}\gf_{q^2}$ is odd. In this case, $S:=\{a\in\gf_q^{*}: \eta\left(a^2-4\right)=-1 \}$, where $\eta(\cdot)$ is  the quadratic character.  By Corollary \ref{Maintheorem_r=1}, we obtain
$$R(a)^2-4=\left(a^2-4\right)\left(\frac{h_1^2(a)-h_2^2(a)}{h_1(a)h_2(a)a+h_1^2(a)+h_2^2(a)}\right)^2.$$
Hence $R(a)$  is a mapping from $S$ to $S$  if $h_1^2(a)-h_2^2(a)\neq0$ for any $a\in S$.  Now we have the following result as a consequence of Theorem \ref{Maintheorem}. We omit all the details because it suffices to  verify all the conditions in Theorem \ref{Maintheorem}, which is similar to Theorem \ref{MainTheorem_even_1}. 

\begin{Th}
	\label{MainTheorem_odd_1}
	Let $q=p^k$, $p$ be an odd prime, $S:=\{a\in\gf_q^{*}: \eta\left(a^2-4\right)=-1 \}$ and $f(x)=xh\left(x^{q-1}\right)\in\gf_{q^2}[x]$ such that $h(x) \in \gf_q[x]$.  Let $a = x+x^{-1}$. By Algorithm \ref{alg1}, we get $h_1(a)$ and $h_2(a)$ from $h(x)=h_1(a)x+h_2(a)$. Assume that $$\psi(a)=\frac{h_1^2(a)-h_2^2(a)}{h_1(a)h_2(a)a+h_1^2(a)+h_2^2(a)}.$$ Then $f(x)$ permutes $\gf_{q^2}$ if and only if the following conditions hold:
	\begin{enumerate}[(1)]
		\item $h(1)\neq0$ and $h(-1)\neq0$;
		\item $h_1^2(a)\neq h_2^2(a)$ for any $a\in S$;
		\item $R(a)=\left(a^2-4\right)\psi(a)^2+4$ permutes $S$.
	\end{enumerate}
\end{Th}

Let $b=a^2-4$ and  $T:=\{b\in\gf_q: \eta(b)=-1, \eta(b+4)=1\}$. Then $R(a)$ permutes $S$ if  and only if $L(b)=bl(b)^2$ permutes $T$, where $l(b)=\frac{h_1^2(\epsilon(b))-h_2^2(\epsilon(b))}{h_1(\epsilon(b))h_2(\epsilon(b))\epsilon(b)+h_1^2(\epsilon(b))+h_2^2(\epsilon(b))}$ and $\epsilon^2(b)=b+4$. Therefore, we rewrite the above result as follows:

\begin{Th}
	\label{MainTheorem_odd_2}
	Let $q=p^k$, $p$ be an odd prime, $T:=\{b\in\gf_q: \eta(b)=-1, \eta(b+4)=1\}$ and $f(x)=xh\left(x^{q-1}\right)\in\gf_q[x]$. By Algorithm \ref{alg1}, we get $h_1(a)$ and $h_2(a)$ from $h(x)=h_1(a)x+h_2(a)$. Assume that $L(b)=bl(b)^2,$ where $l(b)=\frac{h_1^2(\epsilon(b))-h_2^2(\epsilon(b))}{h_1(\epsilon(b))h_2(\epsilon(b))\epsilon(b)+h_1^2(\epsilon(b))+h_2^2(\epsilon(b))}$ and $\epsilon^2(b)=b+4$.  Then $f(x)$ permutes $\gf_{q^2}$ if and only if the following conditions hold:
	\begin{enumerate}[(i)]
		\item $h(1)\neq0$ and $h(-1)\neq0$;
		\item $h_1^2(\epsilon(b))\neq h_2^2(\epsilon(b))$ for any $b\in T$;
		\item $L(b)$ permutes $T$.
	\end{enumerate}
\end{Th}

Given $l(b)$ which permutes $T$, we can obtain $\psi(a)=l\left(a^2-4\right)$. Assume $h_1(a),h_2(a)\neq0$ for any $a\in S$. Then we have $$\psi(a)=\frac{H^2(a)-1}{H^2(a)+aH(a)+1},$$ where $H(a)=\frac{h_1(a)}{h_2(a)}.$ Furthermore, we have 
\begin{equation}
\label{obtain_H(a)}
(\psi(a)-1)H^2(a)+a\psi(a)H(a)+\psi(a)+1=0.
\end{equation}

The discriminant of the above equation is
$$\Delta=\left(a^2-4\right)\psi^2(a)+4.$$
Therefore, we hope that it is not difficult to obtain the square root of $\Delta$ when we construct permutation polynomials from $l(b)$ which permutes $T$. If we can obtain  $H(a)=\frac{t_1(a)}{t_2(a)}$, without loss of generality, we may assume $h_1(a)=t_1(a)$ and $h_2(a)=t_2(a)$. Furthermore, we can obtain $h(x)$ according to Algorithm \ref{alg}. In the following, we demonstrate the construction of some permutation polynomials over $\gf_{q^2}$ from monomials $l(b)$. 

\subsection{The case $l(b)$ is a monomial}

In this subsection, we mainly consider the case that $l(b)=b^s$  such that $L(b)=bl(b)^2$ permutes $T$ in Theorem \ref{MainTheorem_odd_2}. Using  Magma, we obtain some experiment results under the conditions where $k$ is from $2$ to $5$ and $p=3,5,7$, see TABLE II, III and IV, respectively.  All experiment results in these tables can be explained by the following constructions of permutation trinomials. 

\begin{table}[htp]
	\label{table_p=3}
	\centering
	\caption{The value of $s$ such that $L(b)$ permutes $T$, where $\mathrm{char}\gf_{q}=3$}
	\begin{tabular}{ c c c}		
		\hline 
		\hline
		& $k$ &  $s$   \\
		\hline
		& $2$ & $1,4,5$ \\
		\hline
		& $3$ & $1,4,13,14,17$ \\
		\hline
		& $4$  &  $1,4,13,40,41,44,53$   \\
		\hline
		&   $5$   &  $1,4,13,40,121,122,125,134,161$    \\
		\hline
		\hline
	\end{tabular}
\end{table}	

\begin{table}[htp]
	\label{table_p=5}
	\centering
	\caption{The value of $s$ such that $L(b)$ permutes $T$, where $\mathrm{char}\gf_{q}=5$}
	\begin{tabular}{ c c c}		
		\hline 
		\hline
		& $k$ &  $s$   \\
		\hline
		& $2$ & $2,12,14$ \\
		\hline
		& $3$ & $2,12,62,64,74$ \\
		\hline
		& $4$  &  $2,12,62,312,314,324,374$   \\
		\hline
		&   $5$   &  $2,12,62,312,1562,1564,1574,1624,1874$    \\
		\hline
		\hline
	\end{tabular}
\end{table}	

\begin{table}[htp]
	\label{table_p=7}
	\centering
	\caption{The value of $s$ such that $L(b)$ permutes $T$, where $\mathrm{char}\gf_{q}=7$}
	\begin{tabular}{ c c c}		
		\hline 
		\hline
		& $k$ &  $s$   \\
		\hline
		& $2$ & $3,24,27$ \\
		\hline
		& $3$ & $3,24,171,174,195$ \\
		\hline
		& $4$  &  $3,24,171,1200,1203,1224,1371$   \\
		\hline
		\hline
	\end{tabular}
\end{table}	

\begin{Lemma}
	\label{odd_mon_1}
	Let $q=p^k$, $p$ be an odd prime and $T:=\{b\in\gf_q: \eta(b)=-1, \eta(b+4)=1\}$. Let $l(b)=b^s$ and $s=\frac{p^m-1}{2}$, where $m>0$ is an integer. Then $L(b)=bl(b)^2$ permutes $T$.
\end{Lemma}

\begin{proof}
	In the case, $L(b)=bl(b)^2=b^{p^m}$. It is trivial to verify $L(b)$ permutes $T$ because $\eta(L(b))= \eta(b)$ and $L(b) +4 = b^{p^m} + 4 = (b+4)^{p^m}$ implies that $\eta(L(b) +4) =\eta(b+4)$. 
\end{proof}

In the following, we construct permutation polynomials of the form  $xh\left(x^{q-1}\right)$ over $\gf_{q^2}$ from $l(b)$'s satisfying Lemma \ref{odd_mon_1}. In this case, $l(b)=b^{\frac{p^m-1}{2}}$, where $m>0$ is an integer. Then $\psi(a)=l\left(a^2-4\right)=\left(a^2-4\right)^{\frac{p^m-1}{2}}$. Moreover, the discriminant of Eq. (\ref{obtain_H(a)}) is 
\begin{eqnarray*}
	\Delta &=& \left(a^2-4\right)\psi^2(a)+4\\
	&=&\left(a^2-4\right)^{p^m}+4\\
	&=&a^{2p^m}-4^{p^m}+4\\
	&=&a^{2p^m}.
\end{eqnarray*}
The last step is due to the Fermat's little theorem, i.e., $4^{p}\equiv4\pmod p$. Furthermore, we have $$H(a)=\frac{-a\psi(a)\pm a^{p^m}}{2\psi(a)-2}.$$

{\bfseries Case I:} $H(a)=\frac{-a\psi(a)+ a^{p^m}}{2\psi(a)-2}=\frac{h_1(a)}{h_2(a)}$. We let $h_1(a)=-a\psi(a)+ a^{p^m}=-a\left(a^2-4\right)^{\frac{p^m-1}{2}}+a^{p^m}$ and $h_2(a)=2\left(a^2-4\right)^{\frac{p^m-1}{2}}-2.$ 

Let $n$ be even and $U_n(x)=\sum_{i=0}^{\frac{n}{2}}x^{2i}=x^n+x^{n-2}+\cdots+1$. Then according to Algorithm \ref{alg} and using $x^{q} = x^{-1}$ over $\mu_{q+1}$, we can obtain $$h(x)=-x(x+x^q)\left(x-x^q\right)^{p^m-2}+2\left(x-x^q\right)^{p^m-2}+xU_{p^m-1}-x^{q(p^m-2)}U_{p^m-3}.$$
Then we have the following theorem.

\begin{Th}
	\label{odd_permutation}
	Let $q=p^k$, $p$ be a prime, $U_n(x)=x^n+x^{n-2}+\cdots+1$, where $n$ is even, and let $f(x)=xh\left(x^{q-1}\right)$, where $h(x)=-x(x+x^q)\left(x-x^q\right)^{p^m-2}+2\left(x-x^q\right)^{p^m-2}+xU_{p^m-1}-x^{q(p^m-2)}U_{p^m-3}$ and $m>0$ is an integer.  Then $f(x)$ is a permutation polynomial over $\gf_{q^2}$ if and only if  $\gcd\left(\frac{p^m-1}{2}, p^k-1\right)=\gcd\left(\frac{p^m-1}{2}, \frac{p^k-1}{2}\right)$.
\end{Th}

\begin{proof}
	Let $h(x)=h_1(a)x+h_2(a)$. Then it follows from the above statement that $H(a)=\frac{h_1(a)}{h_2(a)}=\frac{-a\psi(a)+ a^{p^m}}{2\psi(a)-2}$, where $\psi(a)=\left(a^2-4\right)^{\frac{p^m-1}{2}}$. According to  Condition (2) of  Theorem \ref{MainTheorem_odd_1}, it suffices to consider the necessary and sufficient conditions so that $H(a)\neq1,-1$ for any $a\in S$. In the following, we only consider the equation $H(a)=1$ and the case $H(a)=-1$ is similar. For $$\frac{-a\psi(a)+ a^{p^m}}{2\psi(a)-2}=1,$$
	i.e., 
	$$\left(a^2-4\right)^{\frac{p^m-1}{2}}=(a+2)^{p^m-1}.$$
    Then we have $$\frac{a+2}{a-2}=\epsilon,$$ where $\epsilon \neq 1$, and 
    \begin{equation}
    \label{ep_m}
    \epsilon^{\frac{p^m-1}{2}}=1.
    \end{equation} Moreover, $a=\frac{2\epsilon+2}{\epsilon-1}.$ Since $$a^2-4=\frac{4^2}{(\epsilon-1)^2}\epsilon,$$ we have  $\eta(\epsilon)=-1$, i.e., 
    \begin{equation}
    \label{ep_k}
     \epsilon^{\frac{p^k-1}{2}}=-1.
    \end{equation}
     Therefore, $H(a)=1$ has a  solution in $S$ if and only if Eq. (\ref{ep_k}) and Eq. (\ref{ep_m}) have at least one common solution  in $\gf_{q}$. On one hand, the number of solutions to Eq. (\ref{ep_m}) in $\gf_{q}$ is $\gcd\left(\frac{p^m-1}{2}, p^k-1\right)$.  On the other hand,  $\epsilon^{\frac{p^k-1}{2}}=-1$ or $1$ for any $\epsilon$ which satisfies Eq. (\ref{ep_m}). This implies that the solutions of Eq. (\ref{ep_m}) all satisfy $\epsilon^{\frac{p^k-1}{2}}=1$ if and only if $\gcd\left(\frac{p^m-1}{2}, p^k-1\right)=\gcd\left(\frac{p^m-1}{2}, \frac{p^k-1}{2}\right)$. Therefore, Eq. (\ref{ep_k}) and Eq. (\ref{ep_m}) have no common solution in $\gf_{q}$ if and only if  $\gcd\left(\frac{p^m-1}{2}, p^k-1\right)=\gcd\left(\frac{p^m-1}{2}, \frac{p^k-1}{2}\right)$.
\end{proof}

{\bfseries Case II:} $H(a)=\frac{-a\psi(a)- a^{p^m}}{2\psi(a)-2}=\frac{h_1(a)}{h_2(a)}$. We have $h_1(a)=-a\psi(a)- a^{p^m}=-a\left(a^2-4\right)^{\frac{p^m-1}{2}}-a^{p^m}$ and $h_2(a)=2\left(a^2-4\right)^{\frac{p^m-1}{2}}-2.$ Then $$h(x)=-x(x+x^q)\left(x-x^q\right)^{p^m-1}+2\left(x-x^q\right)^{p^m-1}-x\left(x+x^q\right)^{p^m}-2.$$ 

Therefore, we can construct the following class of permutation polynomials over $\gf_{q^2}$. The proof is similar to that of Theorem \ref{odd_permutation}, so we omit all the details. 

\begin{Th}
	Let $q=p^k$, $p$ be a prime and $f(x)=xh\left(x^{q-1}\right)$, where $h(x)=-x(x+x^q)\left(x-x^q\right)^{p^m-1}+2\left(x-x^q\right)^{p^m-1}-x\left(x+x^q\right)^{p^m}-2$ and $m>0$ is an integer. Then $f(x)$ is a permutation polynomial over $\gf_{q^2}$ if and only if  $\gcd\left(\frac{p^m-1}{2}, p^k-1\right)=\gcd\left(\frac{p^m-1}{2}, \frac{p^k-1}{2}\right)$.
\end{Th}

\begin{Lemma}
	\label{odd_mon_2}
	Let $q=p^k$, $p$ be an odd prime and $T:=\{b\in\gf_q: \eta(b)=-1, \eta(b+4)=1\}$. Let $l(b)=b^s$ and $s=\frac{p^k+p^m-2}{2}$, where $m>0$ is an integer. Then $L(b)=bl(b)^2$ permutes $T$.
\end{Lemma}

\begin{proof}
	The proof is similar to Lemma \ref{odd_mon_1} and we omit all the details. 
\end{proof}

Similarly, we can also obtain the following two classes of permutation polynomials over $\gf_{q^2}$ from Lemma \ref{odd_mon_2} following the same construction process. 

\begin{Th}
	\label{odd_permutation_2}
	Let $q=p^k$, $p$ be a prime, $U_n(x)=x^n+x^{n-2}+\cdots+1$, where $n$ is even, and $f(x)=xh\left(x^{q-1}\right)$, where $h(x)=x(x+x^q)\left(x-x^q\right)^{p^m-2}-2\left(x-x^q\right)^{p^m-2}+xU_{p^m-1}-x^{q(p^m-2)}U_{p^m-3}$ and $m>0$ is an integer. Then $f(x)$ is a permutation polynomial over $\gf_{q^2}$ if and only if  $\gcd\left(\frac{p^k+p^m-2}{2}, p^k-1\right)=\gcd\left(\frac{p^k+p^m-2}{2}, \frac{p^k-1}{2}\right)$.
\end{Th}

\begin{Th}
	Let $q=p^k$, $p$ be a prime and $f(x)=xh\left(x^{q-1}\right)$, where $h(x)=x(x+x^q)\left(x-x^q\right)^{p^m-1}-2\left(x-x^q\right)^{p^m-1}-x\left(x+x^q\right)^{p^m}-2$ and $m>0$ is an integer. Then $f(x)$ is a permutation polynomial over $\gf_{q^2}$ if and only if  $\gcd\left(\frac{p^k+p^m-2}{2}, p^k-1\right)=\gcd\left(\frac{p^k+p^m-2}{2}, \frac{p^k-1}{2}\right)$.
\end{Th}

\section{Comparison with known results}

Recently, there are many new constructions of permutation trinomials over finite fields, especially with even characteristic. In this section, we consider how our method can explain these earlier results. First of all, by Algorithm \ref{alg1} and our theorems, we compute all permutation trinomials over finite fields with even characteristic  such that the degrees of corresponding fractional polynomials are low. For example, in \cite{LikangquanConstructed}, Li et al. obtained the following theorem.
\begin{Th}
	\label{Example}
	\cite{LikangquanConstructed}
	Let $q=2^k$, $k$ be even, $l$ be an integer such that $\gcd(2l+1,q-1)=1$. Then $f(x)=x^{lq+l+1}+x^{(l+3)q+l-2}+x^{(l-1)q+l+2}$ is a permutation trinomial over $\gf_{q^2}$.
\end{Th} 
In the above theorem, $f(x)=x^rh\left(x^{q-1}\right)$, where $r=1+l(q+1)$, i.e., $1$, and $h(x)=1+x^3+x^{-1}$. Through Algorithm \ref{alg1}, we can obtain $h_1(a)=a^2$ and $h_2(a)=1$ in the case. Furthermore, in Theorem \ref{MainTheorem_even_2}, $l(b)=\frac{h_1\left(\frac{1}{b}\right)}{h_1\left(\frac{1}{b}\right)+h_2\left(\frac{1}{b}\right)}=\frac{1}{1+b^2}$. Therefore, according to Theorem \ref{MainTheorem_even_2}, it suffices to prove that $L(b)=b+\frac{1}{1+b^2}+\frac{1}{1+b^4}$ permutes $T:=\{b\in\gf_{q}: \tr(b)=1 \}$ when $k$ is even. Let $y=1+b$. Then $\tr(y)=\tr(1+b)=\tr(b)=1$ since $k$ is even, which means $y\in T$. Moreover, $L(1+y)=1+y+\frac{1}{y^2}+\frac{1}{y^4}$. Therefore, it suffices to prove that $G(y)=y+\frac{1}{y^2}+\frac{1}{y^4}$ permutes $T$, which is true by Lemma \ref{s_-2}. What amazes us is that though Theorem \ref{PP_s=-2} and Theorem \ref{Example} look very different, they can be obtained from a similar permutation over $T$ by our method. Similarly, we can derive many recent constructions which are listed in the following table. Because the explanation is similar, the details are thus omitted.  It can be seen from this table that most of these known permutation trinomials over finite fields with even characteristic can be explained by our method. In the following table, the column labelled with $``g(x)"$ refers to the corresponding fractional polynomials for the class of permutation trinomials in  references listed in the column with label ``Ref". Their corresponding $L(b)$'s over $T$ are shown in the second column. And the conditions when the corresponding $L(b)$ permutes $T$ are given in the third column,   while the source of the proof (either by a lemma in this paper or  a standard argument in Subsection B.3 of Section 3)  is in the final column.  Lastly,  the symbol $"-"$ means that these cases can not be easily explained by the new method up to now.

\begin{table}[htp]
	\centering
	\caption{The list of known results}
	\begin{tabular}{ c  c  c c  c}		
		\hline 
		\hline
		$g(x)$ & $L(b)$   & Conditions  &  Ref. &  Obtained by \\
		\hline
		$\frac{x^3+x^2+1}{x^3+x+1}$ & $b+\frac{1}{b}+\frac{1}{b^2}$    &  $k>0$ &\cite{DQ,Zhazhengbang} & Lemma \ref{Case_Mon_s=-1}\\
		\hline
		$\frac{x^4+x^3+x}{x^3+x+1}$ &  $b+b^2+b^4$  &  $\gcd(3,k)=1$ & \cite{RS,LikangquanFFa,LikangquanConstructed} & B.3\\
		\hline 
		$\frac{x^5+x^4+1}{x^5+x+1}$ &  $b+\frac{1}{b}+\frac{1}{b^2}$   & $k>0$ & \cite{Zhazhengbang,LiNian1,LikangquanConstructed}  & Lemma \ref{Case_Mon_s=-1}\\
		\hline 
		$\frac{x^5+x^2+x}{x^4+x^3+1}$ & $b+\frac{1}{1+b^2}+\frac{1}{1+b^4}$  & $k$ even   & \cite{Zhazhengbang,LiNian1,LikangquanConstructed} & Lemma \ref{s_-2}\\
		\hline 
		$\frac{x^5+x^4+x}{x^4+x+1}$ &  $b+\frac{1}{1+b}+\frac{1}{1+b^2}$  &  $k$ even & \cite{RS,LiNian1,LikangquanConstructed} & Lemma \ref{Case_Mon_s=-1}\\
		\hline 
		$\frac{x^6+x^2+x}{x^5+x^4+1}$ & $b+b^2+b^4$   &  $\gcd(3,k)=1$ & \cite{LikangquanConstructed} & B.3\\
		\hline
		$\frac{x^7+x^5+1}{x^7+x^2+1}$ & $b+\frac{1}{b^2}+\frac{1}{b^4}$ &  $k>0$  & \cite{LiNian1} & Lemma \ref{s_-2} \\
		\hline
		$\frac{x^7+x^6+x}{x^6+x+1}$ &  $b+\frac{b^2}{b^3+b+1}+\frac{b^4}{b^6+b^2+1}$   &  $\gcd(3,k)=1$ & \cite{LikangquanCCDS} & $-$ \\
		\hline
		$\frac{x^9+x^3+x}{x^8+x^6+1}$ & $b+\frac{b^3}{b^4+b^3+1}+\frac{b^6}{b^8+b^6+1}$ &  $k\not\equiv0\pmod4$  & \cite{LikangquanCCDS} & $-$\\
		\hline
		\hline
	\end{tabular}
\end{table}	


\section{Conclusion}

Motivated by several recent constructions of permutation trinomials over finite fields with even characteristic, in this paper, we present a new and general approach to constructing permutation polynomials of the form $x^rh\left(x^{q-1}\right)$ over $\gf_{q^2}$, where $q=p^k$ and $h(x)\in\gf_q[x]$ is arbitrary. We transform the problem of proving that $f(x)=x^rh\left(x^{q-1}\right)$ is a permutation polynomial over $\gf_{q^2}$ into that of verifying that the corresponding rational function $R(a)$ permutes $S$ (Theorem \ref{Maintheorem}).  This provides a general way to construct permutation polynomials over $\gf_{q^2}$ from these rational functions $R(a)$ which permutes $S$. Because the fractional polynomials $g(x)=x^rh\left(x^{q-1}\right)$ obtained from the same $R(a)$ are the same for different $r$'s, we  concentrate on constructing permutation polynomials of the form $xh\left(x^{q-1}\right)$ over $\gf_{q^2}$, i.e., the case $r=1$. 

For the case $\mathrm{char}\gf_{q^2}=2$,  we construct permutations of $\gf_{q^2}$ from certain rational functions $L(b)=b+l(b)+l(b)^2$ which permutes $T:=\{b\in\gf_{q}: \tr(b)=1  \}$.  We  demonstrate our method using  some specific cases for $l(b)$, i.e., $l(b)$ is a monomial or a  linearized polynomial. In the case where $l(b)=b^s$ is a monomial, we obtain  experimental results under the conditions where $k$ is from $3$ to $12$ and $s$ is not the power of $2$ such that $L(b)=b+l(b)+l(b)^2$ permutes $T$ by using the Magma (see TABLE I). We characterize five infinite  classes of permutation polynomials that explaining all the data except two sporadic cases in TABLE I.  What impresses us is the case where $l(b)$ is a linearized polynomial.  In the case, the original problem is reduced to verifying that $L(b) = b + l(b) + l(b)^2$ permutes the subfield $\gf_q$. In particular,  we obtain all results where $L(b)=b+l(b)+l(b)^2$ is a monomial or binomial, as well as  several results in the case where $L(b)$ is a trinomial.  We want to emphasize that it seems not  easy to prove these results directly using the fractional approach. {Moreover, our new method  can explain most of the known permutation trinomials, which are in \cite{DQ,LiNian1,LikangquanFFa,LikangquanConstructed,Zhazhengbang,RS} over finite fields with even characteristic (see TABLE V).} As for the case $q=p^k$, where $p$ is odd, we mainly construct four classes of permutation polynomials over $\gf_{q^2}$ from monomials $L(b)=bl(b)^2$ which permute $T=\{b\in\gf_q: \eta(b)=-1, \eta(b+4)=1\}$.


\begin{thebibliography}{(1)}
	\bibitem{AGW} A. Akbary, D. Ghioca, Q. Wang,
	\newblock  On constructing permutations of finite fields. \newblock {\em Finite Fields Appl.},  17(2011), 51-67.
	\bibitem{SM} S. Ball, M. Zieve, \newblock Symplectic spreads and permutation polynomials, \newblock {\em in: Finite Fields and Applications}, in: Lect. Notes Comput. Sci., vol.2948, Springer, Berlin, 2004, 79-88.
	\bibitem{SZM} E.R. Berlekamp, H. Rumsey, G. Solomon, \newblock On the solution of algebraic equations over finite fields, \newblock   {\em Information And Control}, 10(1967), 553-564.
	\bibitem{CD} C. Ding, \newblock Cyclic Codes from some monomials and trinomials, \newblock {\em SIAM J. Discrete Math.}, 27(2013), 1977-1994.
	\bibitem{DJ} C. Ding, J. Yuan, \newblock A family of skew Hadamard difference sets, \newblock  {\em J. Combin. Theory Ser. A.}, 113(2006), 1526-1535.
	\bibitem{DQ} C. Ding, L. Qu ,Q. Wang, J. Yuan, P. Yuan, \newblock Permutation trinomials over finite fields with even characteristic, \newblock {\em SLAM J.  Discrete Math.}, 29(2015),  79-92.
	\bibitem{HD} H. Dobbertin, \newblock Almost perfect nonlinear power functions on $\mathbf{GF}(2^{n})$: the Welch case, \newblock {\em IEEE Trans. Inf. Theory.}, 45(1999), 1271-1275.
	\bibitem{XH1} X. Hou, \newblock Permutation polynomials over finite fields--A survey of recent advances,\newblock {\em Finite Fields Appl.}, 32(2015),  82-119.	
	\bibitem{XH4} X. Hou, \newblock Determination of a type of permutaiton trinomials over finite fields, $\mathbf{II}$, \newblock 
	{\em Finite Fields Appl.}, 35(2015),  16-35.
	\bibitem{LC} Y. Laigle-Chapuy, \newblock Permutation polynomials and applications to coding theory, \newblock {\em Finite Fields Appl.}, 13(2007),  58-70.
	\bibitem{LN} R. Lidl, H. Niederreiter, \newblock Finite Fields, 2nd ed. \newblock   {\em Cambridge Univ. Press, Cambridge,} 1997.
	\bibitem{LP} J.B. Lee, Y.H. Park, \newblock Some permuting trinomials over finite fields, \newblock {\em Acta Math. Sci. (English Ed.)}, 17(1997), no. 3, 250-254. MR1483959 (98i:11104).
	\bibitem{LikangquanFFa} K. Li, L. Qu, X. Chen, \newblock New classes of permutation binomials and permutation trinomials over finite fields, \newblock {\em Finite Fields Appl.}, 43(2017), 69-85.
	\bibitem{LikangquanConstructed} K. Li, L. Qu, C. Li and S. Fu, \newblock  New permutation trinomials constructed from fractional polynomials, \newblock  {\em arXiv: 1605.06216v1}, 2016.
	\bibitem{LikangquanCCDS} K. Li, L. Qu, X. Chen and C. Li, \newblock Permutation polynomials of the form $cx+ \tr_{q^l/q}\left(x^a\right)$ and permutation trinomials over finite fields with even characteristic, \newblock {\em Cryptogr. Commun.}, 2017, 1-24.
	\bibitem{LiNian1} N. Li, T. Helleseth, \newblock Several classes of permutation trinomials from Niho exponents, \newblock {\em Cryptogr. Commun.}, 2016, 1-13.
	\bibitem{LiNian2} N. Li, T. Helleseth, \newblock New permutation trinomials from Niho exponents over finite fields with even characteristic, \newblock {\em arXiv: 1606.03768v1}, 2016.
	\bibitem{LW} P.A. Leonard, K.S. Williams, \newblock Quartics over $\mathbb{GF}\left(2^n\right)$. \newblock {\em Proc. Am. Math. Soc.}, 36(1972), 347-350. 
	\bibitem{PL} Y.H. Park, J.B. Lee, \newblock Permutation polynomials and group permutation polynomials. \newblock {\em Bull. Aust. Math. Soc.}, 63(2001), 67-74.
	\bibitem{RS} R. Gupta, R.K. Sharma, \newblock Some new classes of permutation trinomials over finite fields with even characteristic, \newblock {\em Finite Fields Appl.}, 41(2016), 89-96.
	\bibitem{RSL} R.L. Rivest, A. Shamir, L.M. Aselman, \newblock A method for obtaining digital signatures and public-key cryptosystems, \newblock {\em Comm. ACM.}, 21(1978),  120-126.
	\bibitem{ST} J. Sun, O.Y. Takeshita, \newblock Interleavers for turdo codes using permutation polynomials over integer rings, \newblock {\em IEEE Trans. Inform. Theory.}, 51(2005), 101-119.
	\bibitem{W} K.S. Williams,  \newblock  Note on Cubics over $\mathbf{GF}(2^n)$ and $\mathbf{GF}(3^n)^*$.  \newblock {\em Journal of Number Theory},  7(1975), 361-365.
	\bibitem{WQ} Q. Wang, \newblock Cyclotomic mapping permutation polynomials over finite fields. \newblock {\em in: S.W. Golomb, G. Gong, T. Helleseth, H.-Y. Song (Eds.), Sequences, Subsequences, and Consequences, in: Lect. Notes Comput. Sci., vol. 4893, Springer, Berlin.}  119-128, 2007.
    \bibitem{WB} B. Wu, \newblock  The compositional inverses of linearized permutation binomials over finite fields, \newblock{\em arXiv:	1311.2154v1,} 2013.
	\bibitem{YuanDing} P. Yuan, C. Ding, \newblock Permutation polynomials over finite fields from a powerful lemma, \newblock {\em Finite Fields Appl.}, 17(2011), 560-574.


\bibitem{YuanDing:14}
P. Yuan, C. Ding,  \newblock Further results on permutation polynomials over finite fields, \newblock {\em Finite Fields Appl.}, 27 (2014), 88-103.

	\bibitem{Zieve} M.E. Zieve, \newblock On some permutation polynomials over $\mathbb{F}_{q} $ of the form $x^{r}h(x^{(q-1)/d})$, \newblock {\em Proc. Am. Math. Soc.}, 137(2009),  2209-2216.
	\bibitem{Zhazhengbang} Z. Zha, L. Hu, S. Fan, \newblock Further results on permutation trinomials over finite fields with even characteristic, \newblock {\em Finite Fields Appl.}, 45(2017), 43-52.
	\bibitem{Zheng} Y. Zheng, P. Yuan, D. Pei, \newblock Large classes of permutation polynomials over $\gf_{q^2}$, \newblock {\em Des. Codes Cryptogr.}, 81 (2016), 505-521.
\end{thebibliography}
\end{document}